\documentclass[a4paper,UKenglish]{lipics-v2018}

\usepackage[section]{algorithm}
\usepackage[]{algpseudocode} 
\newcommand{\algparbox}[1]{\parbox[t]{\dimexpr\linewidth-\algorithmicindent}{#1\strut}}
\newcommand{\hStatex}[0]{\vspace{5pt}}
\makeatletter
\newlength{\trianglerightwidth}
\algnewcommand{\LineComment}[1]{\Statex \hskip\ALG@thistlm \(\triangleright\) #1}
\algnewcommand{\LineCommentCont}[1]{\Statex \hskip\ALG@thistlm%
  \parbox[t]{\dimexpr\linewidth-\ALG@thistlm}{\hangindent=\trianglerightwidth \hangafter=1 \strut$\triangleright$ #1\strut}}
\makeatother

\algblock{CLASS}{EndClass}
\algnewcommand\algorithmicclass{\textbf{Class}}
\algnewcommand\algorithmicbegin{\textbf{begin}}
\algnewcommand\algorithmicendclass{\textbf{end}}
\algrenewtext{CLASS}[1]{\algorithmicclass\ \textsl{#1}\ \algorithmicbegin}
\algrenewtext{EndClass}[1]{\algorithmicendclass\ \Comment{\algorithmicclass\ \textsl{#1}}}

\newcommand{\VAR}[1]{\textbf{Var} $#1$}
\newcommand{\Null}{{\Lambda}}%
\newcommand{\Nop}{{\small \textsc{NOP}}}%

\newcommand{\sorder}{\prec_{s}}
\newcommand{\SortedSet}{\textbf{SortedSet}}

\newcommand{\Enum}{\textbf{Enum}}
\newcommand{\LR}{\textsl{LG}}
\newcommand{\Ordering}{\textsl{Ord}}
\newcommand{\Left}{\textsl{\small LT}}	 
\newcommand{\Right}{\textsl{\small GT}} 
\newcommand{\SID}{\textsl{\footnotesize SID}}
\newcommand{\CID}{\textsl{\footnotesize CID}}
\newcommand{\RID}{\textsl{\footnotesize RID}}
\newcommand{\SEQ}{\textsl{\footnotesize SEQ}}
\newcommand{\Remote}{\textsl{\footnotesize GLOBAL}}
\newcommand{\Local}{\textsl{\footnotesize LOCAL}}
\newcommand{\opset}{\mathcal{O}}
\newcommand{\Op}{\textsl{Op}}
\newcommand{\Val}{\textsl{Val}}
\newcommand{\Vertex}{\textsl{Vertex}}
\newcommand{\Edge}{\textsl{Edge}}
\newcommand{\StateSpace}{\textsl{StateSpace2D}}
\newcommand{\CStateSpace}{\textsl{CStateSpace}}

\newcommand{\ldel}{\textsc{Del}}	
\newcommand{\lins}{\textsc{Ins}}	

\newcommand{\nat}{\mathbb{N}}
\newcommand{\set}[1]{\{#1\}}
\newcommand{\emptyseq}{\langle \rangle}
\newcommand{\seq}[1]{\langle #1 \rangle}

\newcommand{\doe}{\text{do}}    
\newcommand{\doop}[2]{\text{do}\big(#1, #2\big)}    
\newcommand{\send}{\text{send}}
\newcommand{\sendop}[1]{\text{send}(#1)}
\newcommand{\rcv}{\text{receive}}
\newcommand{\rcvop}[1]{\text{receive}(#1)}

\newcommand{\prot}{$\mathcal{R}$}   
\newcommand{\protinmath}{\mathcal{R}}   
\newcommand{\proj}[2]{#1|_{#2}} 

\newcommand{\instype}{\textsc{Ins}}
\newcommand{\ins}[2]{\textsc{Ins}(#1,#2)}
\newcommand{\deltype}{\textsc{Del}}
\newcommand{\del}[2]{\textsc{Del}(#1,#2)}
\newcommand{\delone}[1]{\textsc{Del}(#1)}
\newcommand{\readtype}{\textsc{Read}}

\newcommand{\placeholder}{\_}

\newcommand{\precrel}[3]{#2 \prec_{#1} #3}
\newcommand{\hbrel}[3]{#2 \xrightarrow{\text{hb}_{#1}} #3}  
\newcommand{\tbrel}[3]{#2 \xrightarrow{\text{tb}_{#1}} #3}  
\newcommand{\vis}{\text{vis}}
\newcommand{\viseq}{\leq_{\text{vis}}}
\newcommand{\visrel}[2]{#1 \xrightarrow{\text{vis}} #2}

\newcommand{\spec}{\mathcal{S}}
\newcommand{\cp}{\mathcal{A}_{\text{cp}}}

\newcommand{\wlspec}{\mathcal{A}_{\text{weak}}}

\newcommand{\elems}[1]{\text{elems}(#1)}
\newcommand{\lo}{\text{lo}}
\newcommand{\lorel}[2]{#1 \xrightarrow{\lo{}} #2}   

\newcommand{\crel}[3]{#2 \xrightarrow{#1} #3} 
\newcommand{\pararel}[3]{#2 \parallel_{#1} #3}  

\newcommand{\pathplain}[2]{$P_{#1 \leadsto #2}$}
\newcommand{\pathinmath}[2]{P_{#1 \leadsto #2}}

\theoremstyle{plain}
\newtheorem{prop}[theorem]{Proposition}
\newtheorem*{claim}{Claim}

\newenvironment{subproof}[1][\proofname]{%
  \begin{proof}[#1]%
}{%
  \end{proof}%
}

\newcommand{\cop}[2]{$#1\set{#2}$}  
\newcommand{\copinmath}[2]{#1\set{#2}}
\newcommand{\opot}[2]{#1\langle#2\rangle}

\newcommand{\jupiter}{\text{Jupiter}}
\newcommand{\cjupiter}{\text{CJupiter}}

\newcommand{\css}{\text{CSS}}

\newcommand{\cscws}[1]{$\text{DSS}_{s_{#1}}$}
\newcommand{\cscwsk}[2]{$\text{DSS}_{s_{#1}}^{#2}$}
\newcommand{\cscwskinmath}[2]{\text{DSS}_{s_{#1}}^{#2}}
\newcommand{\cscwc}[1]{$\text{DSS}_{c_{#1}}$} 
\newcommand{\cscwck}[2]{$\text{DSS}_{c_{#1}}^{#2}$} 
\newcommand{\cscwckinmath}[2]{\text{DSS}_{c_{#1}}^{#2}} 

\newcommand{\csss}{$\text{CSS}_{s}$}
\newcommand{\csssk}[1]{$\text{CSS}_{s}^{#1}$}
\newcommand{\cssskinmath}[1]{\text{CSS}_{s}^{#1}}
\newcommand{\cssc}[1]{$\text{CSS}_{c_{#1}}$}  
\newcommand{\cssck}[2]{$\text{CSS}_{c_{#1}}^{#2}$}  
\newcommand{\cssckinmath}[2]{\text{CSS}_{c_{#1}}^{#2}}  

\usepackage{microtype} 

\title{Specification and Implementation of Replicated List: The Jupiter Protocol Revisited}

\titlerunning{Specification and Implementation of Replicated List: The Jupiter Protocol Revisited} 

\author{Hengfeng Wei}
{State Key Laboratory for Novel Software Technology, Nanjing University, China}
{hfwei@nju.edu.cn}
{https://orcid.org/0000-0002-0427-9710} 
{} 

\author{Yu Huang}
{State Key Laboratory for Novel Software Technology, Nanjing University, China}
{yuhuang@nju.edu.cn}
{https://orcid.org/0000-0001-8921-036X} 
{Contact Author.} 

\author{Jian Lu}
{State Key Laboratory for Novel Software Technology, Nanjing University, China}
{lj@nju.edu.cn}
{} 
{} 

\authorrunning{H. Wei, Y. Huang, and J. Lu} 

\Copyright{Hengfeng Wei, Yu Huang, and Jian Lu} 

\subjclass{ 
  ,
  ,
}


\category{} 

\relatedversion{A short version has been published as a brief announcement by PODC'2018.
} 

\supplement{}

\funding{}

\acknowledgements{} 

\nolinenumbers 

\begin{document}

\maketitle

\begin{abstract}
  The replicated list object is frequently used to model the core functionality
  of replicated collaborative text editing systems.
  Since 1989, the convergence property has been
  a common specification of a replicated list object.
  Recently, Attiya et al. proposed the strong/weak list specification
  and conjectured that the well-known Jupiter protocol
  satisfies the weak list specification.
  The major obstacle to proving this conjecture
  is the mismatch between the global property on all replica states 
  prescribed by the specification
  and the local view each replica maintains in Jupiter
  using data structures like 1D buffer or 2D state space. 
  To address this issue, 
  we propose CJupiter (Compact Jupiter) based on
  a novel data structure called $n$-ary ordered state space
  for a replicated client/server system with $n$ clients.
  At a high level, CJupiter maintains only a single $n$-ary ordered state space
  which encompasses exactly all states of each replica.
  We prove that CJupiter and Jupiter are equivalent and that
  CJupiter satisfies the weak list specification,
  thus solving the conjecture above.
\end{abstract}

\clearpage
\section{Introduction}  \label{section:intro}

Collaborative text editing systems, like Google Docs~\cite{GoogleDocs}, 
Apache Wave~\cite{Wave}, or wikis~\cite{Leuf:Wiki01},
allows multiple users to concurrently edit the same document.
For availability, such systems often replicate the document at several \emph{replicas}.
For low latency, replicas are required to respond to user operations immediately
without any communication with others
and updates are propagated asynchronously.

The \emph{replicated list object} has been frequently used to model the core functionality
(e.g., insertion and deletion) of replicated collaborative text editing 
systems~\cite{Ellis:SIGMOD89, Nichols:UIST95, Xu:CSCW14, Attiya:PODC16}.
A common specification of a replicated list object
is the \emph{convergence} property, proposed by Ellis et al.~\cite{Ellis:SIGMOD89}.
It requires the \emph{final lists at all replicas}
be identical after executing the same set of user operations.
Recently, Attiya et al.~\cite{Attiya:PODC16} proposed the strong/weak list specification.
Beyond the convergence property, the strong/weak list specification specifies global properties 
on \emph{intermediate states} going through by replicas.
Attiya et al.~\cite{Attiya:PODC16} have proved that the existing RGA protocol~\cite{Roh:JPDC11}
satisfies the strong list specification.
Meanwhile, it is \emph{conjectured} that the well-known \jupiter{} protocol~\cite{Nichols:UIST95, Xu:CSCW14},
which is behind Google Docs~\cite{GoogleDocsJupiter} and Apache Wave~\cite{WaveDoc},
satisfies the weak list specification.

\jupiter{} adopts a \emph{centralized server} replica for
propagating updates~\footnote{Since replicas are required to respond to user operations immediately,
the client/server architecture does not imply that clients process operations in the same order.},
and client replicas are connected to the server replica via FIFO channels;
see Figure~\ref{fig:jupiter-schedule-podc16}~\footnote{The details about Figure~\ref{fig:jupiter-schedule-podc16}
will be described in Examples~\ref{ex:cjupiter} and~\ref{ex:jupiter}.}.
\jupiter{} relies on the technique of 
operational transformations (OT)~\cite{Ellis:SIGMOD89, Sun:CSCW98} to achieve convergence.
The basic idea of OT is for each replica to execute any local operation immediately
and to transform a remote operation so that it takes into account 
the concurrent operations previously executed at the replica.
Consider a replicated list system consisting of replicas $R_1$ and $R_2$
which initially hold the same list (Figure~\ref{fig:ex-ot}).
Suppose that user 1 invokes $o_1 = \ins{f}{1}$ at $R_1$
and concurrently user 2 invokes $o_2 = \delone{5}$ at $R_2$.
After being executed locally, each operation is sent to the other replica.
Without OT (Figure~\ref{fig:ex-without-ot}), the states of two replicas diverge.
With the OT of $o_1$ and $o_2$ (Figure~\ref{fig:ex-with-ot}),
$o_2$ is transformed to $o_2' = \delone{6}$ at $R_1$,
taking into account the fact that $o_1$ has inserted an element at position $1$.
Meanwhile, $o_1$ remains unchanged.
As a result, two replicas converge to the same list.
We note that although the idea of OT is straightforward, 
many OT-based protocols for replicated list are hard to understand
and some of them have even been shown incorrect 
with respect to convergence~\cite{Ellis:SIGMOD89, Sun:CSCW98, Sun:CSCW14}.

The major obstacle to proving that \jupiter{} satisfies the weak list specification
is the \emph{mismatch} between the \emph{global property} on all states
prescribed by such a specification
and the \emph{local view} each replica maintains in the protocol.
On the one hand, the weak list specification
requires that states across the system are pairwise compatible~\cite{Attiya:PODC16}.
That is, for any pair of (list) states, there cannot be two elements $a$ and $b$ such that
$a$ precedes $b$ in one state but $b$ precedes $a$ in the other.
On the other hand, Jupiter uses data structures like 1D buffer~\cite{Shen:CSCW02}
or 2D state space~\cite{Nichols:UIST95, Xu:CSCW14}
which are not ``compact'' enough to capture all replica states in one.
In particular, \jupiter{} maintains $2n$ 2D state spaces for a system with $n$ clients~\cite{Xu:CSCW14}:
Each client maintains a single state space which is synchronized with those of other clients
via its counterpart state space maintained by the server.
Each 2D state space of a client (as well as its counterpart at the server) 
consists of a local dimension and a global dimension,
keeping track of the operations processed by the client itself and the others, respectively.
In this way, replica states of \jupiter{} are dispersed in multiple 2D state spaces 
maintained locally at individual replicas.

To resolve the mismatch,
we propose \cjupiter{} (Compact Jupiter), a variant of \jupiter{},
which uses a novel data structure called $n$-\emph{ary ordered state space}
for a system with $n$ clients.
\cjupiter{} is compact in the sense that at a high level, 
it maintains only a single $n$-ary ordered state space
which encompasses exactly all states of each replica.
Each replica behavior corresponds to a path going through this state space.
This makes it feasible for us to reason about global properties
and finally prove that \jupiter{} satisfies the weak list specification,
thus solving the conjecture of Attiya et al.
The roadmap is as follows:
\begin{itemize}
  \item (Section~\ref{section:cjupiter})
    We propose \cjupiter{} based on the $n$-ary ordered state space data structure.
  \item (Section~\ref{section:jupiter-cjupiter-equiv})
    We prove that \cjupiter{} is equivalent to \jupiter{}
    in the sense that the behaviors of corresponding replicas of these two protocols
    are the same under the same schedule of operations.
    \jupiter{} is slightly optimized in implementation at clients (but not at the server) by eliminating redundant OTs,
    which, however, has obscured the similarities among clients and led to the mismatch discussed above.
  \item (Section~\ref{section:cjupiter-weak-spec})
    We prove that \cjupiter{} satisfies the weak list specification.
    Thanks to the ``compactness'' of \cjupiter{},
    we are able to focus on a single $n$-ary ordered state space
    which provides a global view of all possible replica states.
\end{itemize}

\begin{figure}[t]
  \centering
  \begin{minipage}{0.49\textwidth}
    \centering
    \includegraphics[width = 0.75\textwidth]{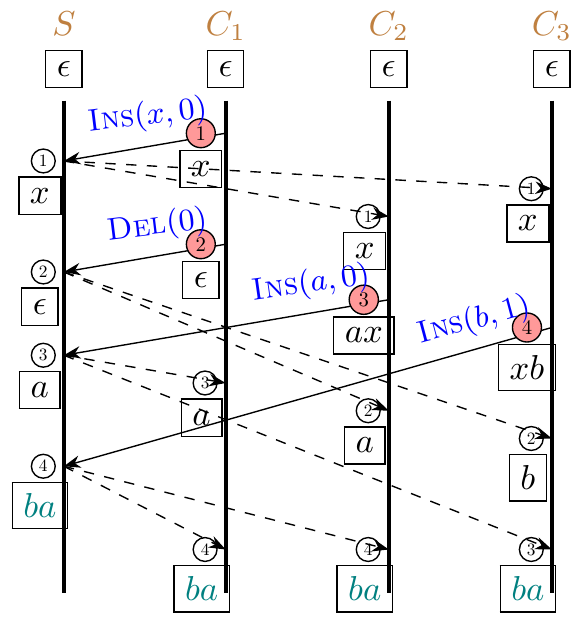}
    \caption{A schedule of four operations adapted from~\cite{Attiya:PODC16}, 
      involving a server replica $s$ and three client replicas $c_1$, $c_2$, and $c_3$.
      The circled numbers indicate the order in which the operations are received at the server.
      The list contents produced by \cjupiter{} (Section~\ref{section:cjupiter}) are shown in boxes.}
    \label{fig:jupiter-schedule-podc16}
  \end{minipage}%
  \hspace{5pt}\hfil
  \begin{minipage}{0.49\textwidth}
    \centering
    \begin{subfigure}[b]{0.44\textwidth}
      \includegraphics[width = \textwidth]{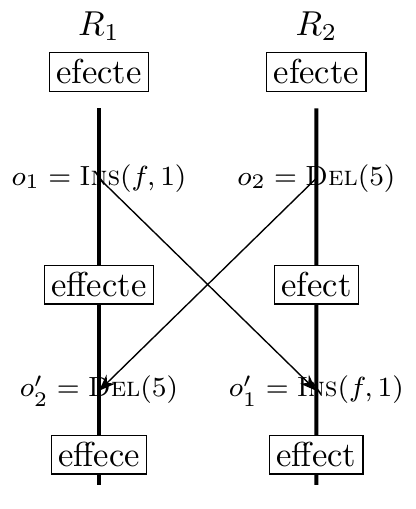}
      \caption{Without OT, the states of $R_1$ and $R_2$ diverge.}
      \label{fig:ex-without-ot}
    \end{subfigure}
    \hfil
    \begin{subfigure}[b]{0.46\textwidth}
      \includegraphics[width = \textwidth]{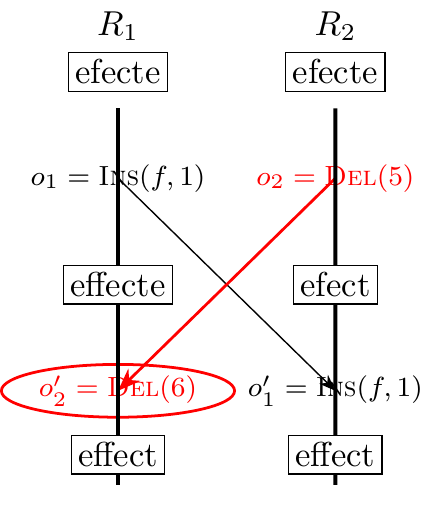}
      \caption{With OT, $R_1$ and $R_2$ converge to the same state.}
      \label{fig:ex-with-ot}
    \end{subfigure}
    \caption{Illustrations of OT (adapted from~\cite{Imine:TCS06}).}
    \label{fig:ex-ot}
  \end{minipage}
\end{figure}

Section~\ref{section:preliminaries} presents preliminaries on specifying replicated list data type and OT.
Section~\ref{section:related-work} describes related work.
Section~\ref{section:conclusion} concludes the paper.
Appendix contains proofs and pseudocode.
\section{Preliminaries: Replicated List and Operational Transformation} \label{section:preliminaries}

We describe the system model and specifications of replicated list
in the framework for specifying replicated data types~\cite{Burckhardt:POPL14, Attiya:PODC15, Attiya:PODC16}.

\subsection{System Model}   \label{ss:model}

A highly-available replicated data store
consists of \emph{replicas} that process user operations on the replicated objects
and communicate updates to each other with messages.
To be \emph{highly-available}, replicas are required to respond to user operations immediately
without any communication with others.
%
  A \emph{replica} is defined as a state machine $R = (\Sigma, \sigma_0, E, \Delta)$, where
  \itshape 1) \upshape $\Sigma$ is a set of states;
  \itshape 2) \upshape $\sigma_0 \in \Sigma$ is the initial state;
  \itshape 3) \upshape $E$ is a set of possible events; and
  \itshape 4) \upshape $\Delta: \Sigma \times E \to \Sigma$ is a transition function.
%
The state transitions determined by $\Delta$ are local steps of a replica,
describing how it interacts with the following three kinds of \emph{events} from users and other replicas:
\begin{itemize}
  \item $\doe(o, v)$: a user invokes an operation $o \in \opset$ on the replicated object
    and immediately receives a response $v \in \Val$.
    We leave the users unspecified and say that the replica \emph{generates} the operation $o$;
  \item $\send(m)$: the replica sends a message $m$ to some replicas; and
  \item $\rcv(m)$: the replica receives a message $m$.
\end{itemize}

A \emph{protocol} is a collection \prot{} of replicas.
An \emph{execution} $\alpha$ of a protocol \prot{} is a sequence of all events
occurring at the replicas in \prot{}.
We denote by $R(e)$ the replica at which an event $e$ occurs.
For an execution (or generally, an event sequence) $\alpha$,
we denote by $\precrel{\alpha}{e}{e'}$
(or $\precrel{}{e}{e'}$)
that $e$ precedes $e'$ in $\alpha$.
%
  An execution $\alpha$ is \emph{well-formed} if
  for every replica $R$:
    \itshape 1) \upshape the subsequence of events $\seq{e_1, e_2, \ldots}$ at $R$,
      denoted $\proj{\alpha}{R}$, is well-formed,
      namely there is a sequence of states $\seq{\sigma_1, \sigma_2, \ldots}$,
      such that $\sigma_i = \Delta(\sigma_{i-1}, e_i)$ for all $i$; and
    \itshape 2) \upshape every $\rcv(m)$ event at $R$ is preceded by a $\send(m)$ event in $\alpha$.
We consider only \emph{well-formed} executions.

We are often concerned with replica behaviors and states when studying a protocol.
%
  The \emph{behavior} of replica $R$ in $\alpha$ is a sequence of the form:
  $\sigma_0, e_1, \sigma_1, e_2, \ldots$,
  where $\seq{e_1, e_2, \ldots} = \proj{\alpha}{R}$ and
  $\sigma_i = \Delta(\sigma_{i-1}, e_i)$ for all $i$.
  A \emph{replica state} $\sigma$ of $R$ in $\alpha$
  can be represented by the events in a prefix of $\alpha|_{R}$ it has processed.
  Specifically, $\sigma_0 = \emptyseq$ and
  $\sigma_{i} = \sigma_{i-1} \circ e_i = \seq{e_1, e_2, \ldots, e_i}$.

We now define the causally-before, concurrent, and totally-before relations on events in an execution.
When restricted to the \doe{} events only, they define relations on user operations.
%
  In an execution $\alpha$, event $e$ is \emph{causally before} $e'$,
  denoted $\hbrel{\alpha}{e}{e'}$ (or $\hbrel{}{e}{e'}$),
  if one of the following conditions holds~\cite{Lamport:CACM78}:
  \itshape 1) \upshape Thread of execution: $R(e) = R(e') \land \precrel{\alpha}{e}{e'}$;
  \itshape 2) \upshape Message delivery: $e = \sendop{m} \land e' = \rcvop{m}$;
  \itshape 3) \upshape Transitivity: $\exists e'' \in \alpha: \hbrel{\alpha}{e}{e''} \land \hbrel{\alpha}{e''}{e'}$.
%
  Events $e, e' \in \alpha$ are \emph{concurrent},
  denoted $\pararel{\alpha}{e}{e'}$ (or $\pararel{}{e}{e'}$),
  if it is neither $\hbrel{\alpha}{e}{e'}$ nor $\hbrel{\alpha}{e'}{e}$.
%
  A relation on events in an execution $\alpha$,
  denoted $\tbrel{\alpha}{e}{e'}$ (or $\tbrel{}{e}{e'}$),
  is a \emph{totally-before} relation \emph{consistent with}
  the causally-before relation `$\hbrel{\alpha}{}{}$' on events in $\alpha$ if
  it is total: $\forall e, e' \in \alpha: \tbrel{\alpha}{e}{e'} \lor \tbrel{\alpha}{e'}{e}$, 
  and it is consistent: $\forall e, e' \in \alpha: \hbrel{\alpha}{e}{e'} \implies \tbrel{\alpha}{e}{e'}$.
\subsection{Specifying Replicated Objects}   \label{ss:specifying}

A replicated object is specified by a set of abstract executions
which record user operations (corresponding to \doe{} events) 
and visibility relations on them~\cite{Burckhardt:POPL14}.
%
  An \emph{abstract execution} is a pair $A = (H, \vis{})$, where
  $H$ is a sequence of \doe{} events and
  $\vis{} \subseteq H \times H$ is an acyclic \emph{visibility} relation such that
    \itshape 1) \upshape if $\precrel{H}{e_1}{e_2}$ and $R(e_1) = R(e_2)$, then $\visrel{e_1}{e_2}$;
    \itshape 2) \upshape if $\visrel{e_1}{e_2}$, then $\precrel{H}{e_1}{e_2}$; and
    \itshape 3) \upshape \vis{} is transitive: $(\visrel{e_1}{e_2} \land \visrel{e_2}{e_3}) \implies \visrel{e_1}{e_3}$.

An abstract execution $A' = (H', \vis{}')$ is a \emph{prefix} of another abstract execution $A = (H, \vis{})$ if
$H'$ is a prefix of $H$ and $\vis{}' = \vis{}\; \cap\; (H' \times H')$.
%
  A \emph{specification} $\spec{}$ of a replicated object is a \emph{prefix-closed} set of abstract executions,
  namely if $A \in \spec{}$, then $A' \in \spec{}$ for each prefix $A'$ of $A$.
%
  A protocol \prot{} \emph{satisfies} a specification $\spec{}$, denoted $\protinmath \models \spec$, if
  any (concrete) execution $\alpha$ of \prot{} \emph{complies with} some abstract execution $A = (H, \vis{})$ in $\spec{}$, namely
  $\forall R \in \protinmath{}: \proj{H}{R} = \proj{\alpha}{R}^{\doe{}}$,
  where $\proj{\alpha}{R}^{\doe{}}$ is the subsequence of \doe{} events of replica $R$ in $\alpha$.

\subsection{Replicated List Specification} \label{ss:list-spec}

A replicated list object supports three types of user operations~\cite{Attiya:PODC16}
($U$ for some universe):
\begin{itemize}
  \item $\textsc{Ins}(a,p)$:
    inserts $a \in U$ at position $p \in \nat{}$ and returns the updated list.
    For $p$ larger than the list size, we assume an insertion at the end.
    We assume that all inserted elements are unique,
    which can be achieved by attaching replica identifiers and sequence numbers.
  \item $\textsc{Del}(a, p)$:
    deletes an element at position $p \in \nat{}$ and returns the updated list.
    For $p$ larger than the list size, we assume an deletion at the end.
    The parameter $a \in U$ is used to record the deleted element~\cite{Sun:CSCW14},
    which will be referred to in condition 1(a) of the weak list specification defined later.
  \item $\textsc{Read}$:
    returns the contents of the list.
\end{itemize}

The operations above, as well as a special \Nop{} (i.e., ``do nothing''), form $\opset$ 
and all possible list contents form $\Val$.
\instype{} and \deltype{} are collectively called \emph{list updates}.
We denote by 
$\elems{A} = \big\{a \mid \doe{}\big(\instype(a, \placeholder{}), \placeholder{}\big) \in H\big\}$
the set of all elements inserted into the list in an abstract execution $A = (H, \vis{})$.

We adopt the convergence property in~\cite{Attiya:PODC16} which requires that
two \readtype{} operations that observe the same set of list updates return the same response.
Formally,
%
  an abstract execution $A = (H, \vis{})$ belongs to the \emph{convergence property} $\cp{}$
  if and only if for any pair of \readtype{} events
  $e_1 = \doop{\readtype}{w_1 \triangleq a_1^0 \ldots a_1^{m-1}}$
  and $e_2 = \doop{\readtype}{w_2 \triangleq a_2^0 \ldots a_2^{n-1}}$
  ($a_i^{j} \in \elems{A}$), it holds that
  $\Big(\vis_{\instype, \deltype}^{-1}(e_1) = \vis_{\instype, \deltype}^{-1}(e_2)\Big) \implies w_1 = w_2$,
  where $\vis_{\instype, \deltype}^{-1}(e)$ denotes the set of list updates visible to $e$.


The weak list specification requires the ordering between elements that are not deleted
to be consistent across the system~\cite{Attiya:PODC16}.

\begin{definition}[Weak List Specification $\wlspec{}$~\cite{Attiya:PODC16}]   \label{def:wl-spec}
An abstract execution $A = (H, \vis{})$ belongs to the \emph{weak list specification} $\wlspec{}$
if and only if there is a relation $\text{lo} \subseteq \elems{A} \times \elems{A}$,
called the \emph{list order}, such that:
\begin{enumerate}
  \item Each event $e = \doe{}(o,w) \in H$ returns a sequence of elements
  $w = a_0 \ldots a_{n-1}$, where $a_i \in \elems{A}$, such that: 
    \begin{enumerate}
      \item $w$ contains exactly the elements visible to $e$ that have been inserted, but not deleted:
	\[
	  \forall a.\, a \in w \iff \Big(\doop{\ins{a}{\placeholder}}{\placeholder} \viseq{} e\Big)
	  \land \lnot \Big(\doop{\del{a}{\placeholder}}{\placeholder} \viseq{} e\Big).
	\]
      \item The list order is consistent with the order of the elements in $w$:
	\[
	  \forall i,j.\, (i < j) \implies (a_i, a_j) \in \lo{}.
	\]
      \item Elements are inserted at the specified position: $op = \ins{a}{k} \implies a = a_{\min\set{k,n-1}}$.
    \end{enumerate}
  \item \lo{} is irreflexive and for all events $e = \doe(op, w) \in H$,
  it is transitive and total on $\set{a \mid a \in w}$.
\end{enumerate}
\end{definition}

\begin{example}[Weak List Specification] \label{ex:wlspec}
  In the execution depicted in Figure~\ref{fig:jupiter-schedule-podc16} (produced by \cjupiter),
  there exist three states with list contents $w_1 = ba$, $w_2 = ax$, and $w_3 = xb$, respectively.
  This is allowed by the weak list specification with the list order $\lo{}$:
  $\lorel{b}{a}$ on $w_1$, $\lorel{a}{x}$ on $w_2$, and $\lorel{x}{b}$ on $w_3$.
  However, an execution is not allowed by the weak list specification 
  if it contained two states with, say $w = ab$ and $w' = ba$. 
\end{example}

\subsection{Operational Transformation (OT)} \label{ss:preliminary-ot}

The OT of transforming $o_1 \in \opset{}$ with $o_2 \in \opset{}$ 
is expressed by the function $o_1' = OT(o_1, o_2)$.
We also write $(o_1', o_2') = OT(o_1, o_2)$
to denote both $o_1' = OT(o_1, o_2)$ and $o_2' = OT(o_2, o_1)$.
To ensure the convergence property, OT functions are required to satisfy CP1
(Convergence Property 1)~\cite{Ellis:SIGMOD89}:
%
  Given two operations $o_1$ and $o_2$, 
  if $(o_1', o_2') = OT(o_1, o_2)$, 
  then $\sigma ; o_1 ; o_2' = \sigma ; o_2 ; o_1'$ should hold,
  meaning that the same state is obtained by applying $o_1$ and $o_2'$ in sequence,
  and applying $o_2$ and $o_1'$ in sequence, on the same initial state $\sigma$.
%
A set of OT functions satisfying CP1 for a replicated list object~\cite{Ellis:SIGMOD89, Imine:TCS06, Sun:CSCW14}
can be found in Figure~\ref{fig:list-ot}. 

\section{The \cjupiter{} Protocol} \label{section:cjupiter}

In this section we propose \cjupiter{} (Compact \jupiter{})
for a replicated list based on the data structure called $n$-ary ordered state space.
Like \jupiter{}, \cjupiter{} also adopts a client/server architecture.
For convenience, we assume that the server does not
generate operations~\cite{Xu:CSCW14, Attiya:PODC16}.
It mainly serializes operations and propagates them from one client to others.
We denote by `$\sorder$' the total order on the set of operations established by the server.
Note that `$\sorder$' is consistent with the causally-before relation `$\hbrel{}{}{}$'.
To facilitate the comparison of \jupiter{} and \cjupiter{},
we refer to `$\hbrel{}{}{}$' and `$\sorder$' together as the {\emph{schedule}} of operations.

\subsection{Data Structure: $n$-ary Ordered State Space}   \label{ss:css-nary}

For a client/server system with $n$ clients,
\cjupiter{} maintains $(n + 1)$ $n$-ary ordered state spaces,
one per replica (\csss{} for the server and \cssc{i} for client $c_i$).
Each \css{} is a directed graph whose vertices represent states 
and edges are labeled with operations;
see Appendix~\ref{section:appendix-css}.

  An \textit{\textbf {operation}} $op$ of type $\Op$ is a tuple $op = (o, oid, ctx, sctx)$, where 
  \itshape 1) \upshape
    $o$ is the signature of type $\opset$ described in Section~\ref{ss:list-spec};
  \itshape 2) \upshape
    $oid$ is a globally unique operation identifier
      which is a pair $(cid, seq)$ consisting of the client id and a sequence number; 
  \itshape 3) \upshape
    $ctx$ is an \emph{operation context} which is a set of $oid$s,
      denoting the operations that are causally before $op$; and
  \itshape 4) \upshape
    $sctx$ is a set of $oid$s, 
    denoting the operations that, as far as $op$ knows, have been executed before $op$ at the server.
    At a given replica, $sctx$ is used to determine the total order `$\prec_{s}$` relation between two operations
    as in Algorithm~\ref{alg:css-op}.

The OT function of two operations $op, op' \in \Op$,
denoted $(\opot{op}{op'}: \Op{}, \opot{op'}{op}: \Op{}) = OT(op, op')$,
is defined based on that of $op.o, op'.o \in \opset$,
denoted $(o, o') = OT(op.o, op'.o)$, such that
$\opot{op}{op'} = (o, op.oid, op.ctx \;\cup\; \set{op'.oid}, op.sctx)$ and
$\opot{op'}{op} = (o', op'.oid, op'.ctx\; \cup\; \set{op.oid}, op'.sctx)$.

  A \textit{\textbf{vertex}} $v$ of type \Vertex{} is a pair $v = (oids, edges)$, where
  $oids$ is the set of operations (represented by their identifies) that have been executed, and
  $edges$ is an \emph{ordered} set of edges of type \Edge{} from $v$ to other vertices, 
  labeled with operations.
  That is, each \textit{\textbf{edge}} is a pair $(op: \Op, v: \Vertex)$.
  Edges from the same vertex are \emph{totally ordered} by their $op$ components.
%
For each vertex $v$ and each edge $e = (op, u)$ from $v$ to $u$, it is required that
\begin{itemize}
  \item the $ctx$ of $op$ associated with $e$ matches the $oids$ of $v$: 
    $op.ctx = v.oids$;
  \item the $oids$ of $u$ consists of the $oids$ of $v$ and the $oid$ of $op$:
    $u.oids = v.oids \cup \set{op.oid}$.
\end{itemize}

\begin{figure}[t]
  \centering
  \begin{minipage}{0.38\textwidth}
    \centering
    \includegraphics[width = 0.70\textwidth]{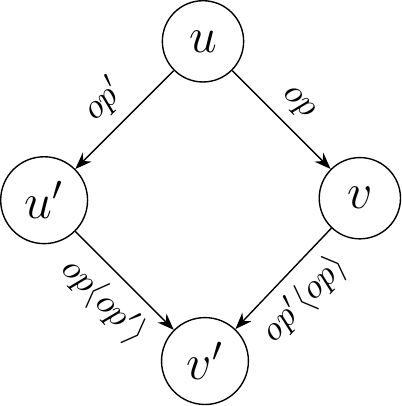}
    \caption{Illustration of an OT of two operations $op, op'$
      in both the $n$-ary ordered state space of \cjupiter{} 
      and the 2D state space of \jupiter{}:
      $(\opot{op}{op'}, \opot{op'}{op}) = OT(op, op')$.
      In the \cjupiter{} and \jupiter{} protocols 
      (and Examples~\ref{ex:cjupiter} and~\ref{ex:jupiter}),
      $op$ corresponds to the new incoming operation to be transformed.}
    \label{fig:xform-ot}
  \end{minipage}
  \hspace{5pt}\hfil
  \begin{minipage}{0.60\textwidth}
    \centering
    \includegraphics[width = 0.80\textwidth]{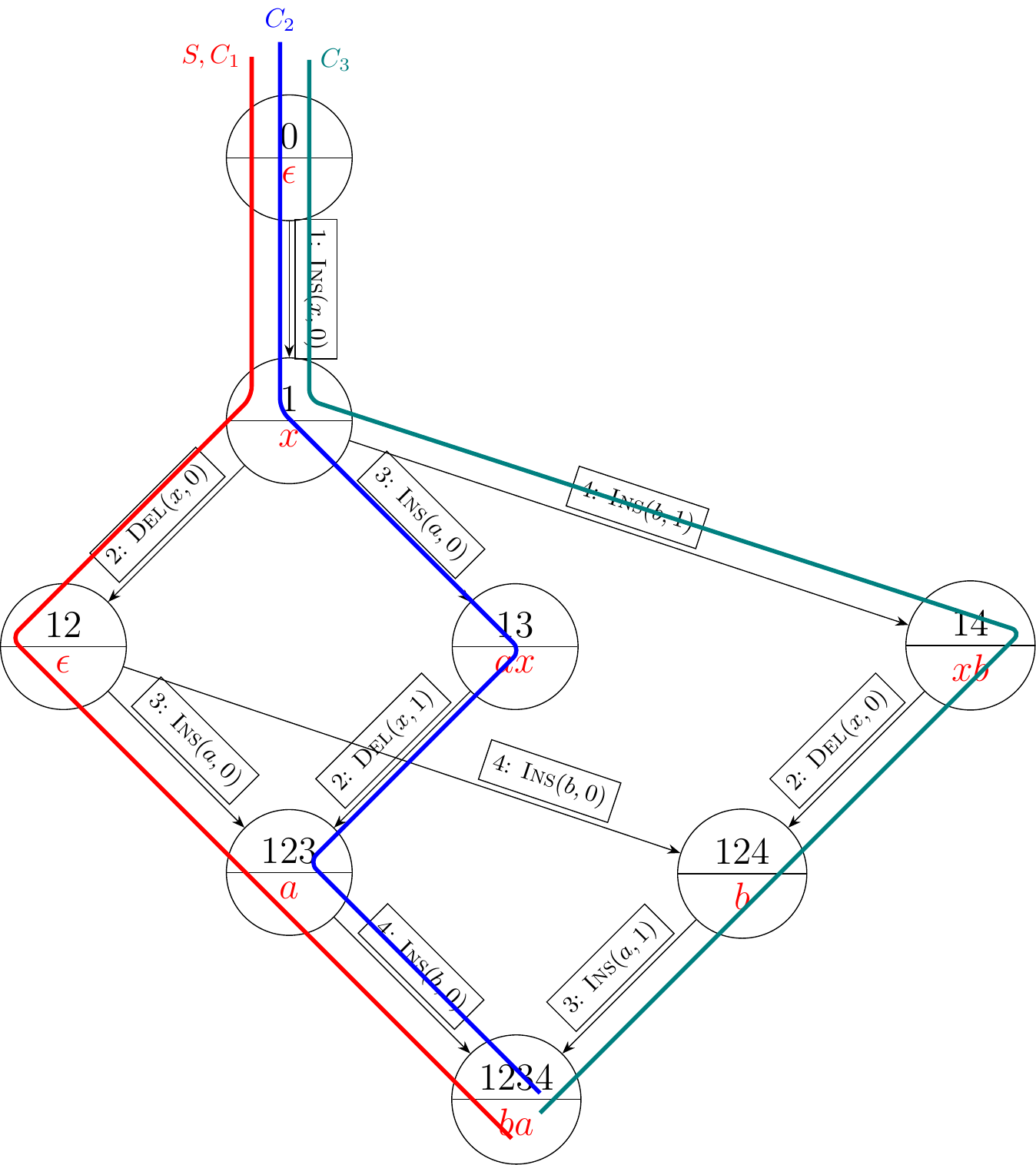}
    \caption{The same final $n$-ary ordered state space (thus for \csss{} and each \cssc{i}) 
    constructed by \cjupiter{} for each replica under the schedule of Figure~\ref{fig:jupiter-schedule-podc16}.
    Each replica behavior (i.e., the sequence of state transitions) corresponds to a path going through this state space.
  }
    \label{fig:cjupiter-css-podc16-allinone}
  \end{minipage}
\end{figure}

\begin{definition}[$n$-ary Ordered State Space]	\label{def:css-nary}
  An \textit{\textbf{$n$-ary ordered state space}} is a set of vertices such that
  \begin{enumerate}
    \item Vertices are uniquely identified by their $oids$.
    \item For each vertex $u$ with $|u.edges| \ge 2$,
      let $u'$ be its child vertex along the \textit{\textbf{first}} edge $e_{uu'} = (op', u')$
      and $v$ another child vertex along $e_{uv} = (op, v)$.  
      There exist (Figure~\ref{fig:xform-ot})
      \begin{itemize}
	\item a vertex $v'$ with $v'.oids = u.oids \cup \set{op'.oid, op.oid}$;
	\item two edges $e_{u'v'} = (\opot{op}{op'}, v')$ from $u'$ to $v'$ and $e_{vv'} = (\opot{op'}{op}, v')$ from $v$ to $v'$.
      \end{itemize}
  \end{enumerate}
\end{definition}

The second condition models OTs in \cjupiter{} described in Section~\ref{ss:cjupiter},
and the choice of the ``first'' edge is justified in Lemmas~\ref{lemma:cjupiter-first-rule} and~\ref{lemma:cjupiter-ot-server}.
\subsection{The \cjupiter{} Protocol}   \label{ss:cjupiter}

Each replica in \cjupiter{} maintains an $n$-ary ordered state space $S$ 
and keeps the most recent vertex $cur$ 
(initially $(\emptyset, \emptyset)$) of $S$.
Following~\cite{Xu:CSCW14}, we describe \cjupiter{} in three parts;
see Appendix~\ref{section:appendix-cjupiter-protocol} for pseudocode.

{\bf Local Processing Part.}
When a client receives an operation $o \in \opset$ from a user, it
\begin{enumerate}
  \item applies $o$ locally, obtaining a new list $val \in \Val$;
  \item generates $op \in \Op$ by attaching to $o$ a unique operation identifier 
    and the operation context $S.cur.oids$, representing the set of operations that are causally before $op$;
  \item creates a vertex $v$ with $v.oids = S.cur.oids \cup \set{op.oid}$,
    appends $v$ to $S$ by linking it to $S.cur$ via an edge labeled with $op$,
    and updates $cur$ to be $v$;
  \item sends $op$ to the server asynchronously and returns $val$ to the user.
\end{enumerate}

{\bf Server Processing Part.}   
To establish the total order `$\prec_{s}$' on operations, 
the server maintains the set $soids$ of operations it has executed.
When the server receives an operation $op \in \Op$ from client $c_i$, it
\begin{enumerate}
  \item updates $op.sctx$ to be $soids$ and updates $soids$ to include $op.oid$;
  \item transforms $op$ with an operation sequence in $S$ to obtain $op'$
    by calling $S.\textsc{xForm}(op)$ (see below), 
    and applies $op'$ (specifically, $op'.o$) locally;
  \item sends $op$ (instead of $op'$) to other clients asynchronously.
\end{enumerate}

{\bf Remote Processing Part.}
When a client receives an operation $op \in \Op$ from the server, 
it transforms $op$ with an operation sequence in $S$ to obtain $op'$ 
by calling $S.\textsc{xForm}(op)$ (see below),
and applies $op'$ (specifically, $op'.o$) locally.


{\bf OTs in \cjupiter{}.}
The procedure $S.\textsc{xForm}(op: \Op)$ transforms $op$
with an operation sequence in an $n$-ary ordered state space $S$.
Specifically, it
\begin{enumerate}
  \item locates the vertex $u$ whose $oids$ matches the $ctx$ of $op$, i.e., $u.oids = op.ctx$~\footnote{
    The vertex $u$ exists due to the FIFO communication between the clients and the server.},
    and creates a vertex $v$ with $v.oids = u.oids \cup \set{op.oid}$;
  \item iteratively transforms $op$ with an operation sequence consisting of operations 
    along the \textit{\textbf{first}} edges from $u$ to the final vertex $cur$ of $S$ (Figure~\ref{fig:xform-ot}):
    \begin{enumerate}
      \item obtains the vertex $u'$ and the operation $op'$ associated with the first edge of $u$;
      \item transforms $op$ with $op'$ to obtain $\opot{op}{op'}$ and $\opot{op'}{op}$;
      \item creates a vertex $v'$ with $v'.oids = v.oids \cup \set{op'.oid}$; 
      \item links $v'$ to $v$ via an edge labeled with $\opot{op'}{op}$
	and $v$ to $u$ via an edge labeled with $op$;
      \item updates $u$, $v$, and $op$ to be $u'$, $v'$, and $\opot{op}{op'}$, respectively;
    \end{enumerate}
  \item when $u$ is the final vertex $cur$ of $S$, links $v$ to $u$ via an edge labeled with $op$,
    updates $cur$ to be $v$, and returns the last transformed operation $op$.
\end{enumerate}


\begin{figure}[t]
  \centering
  \begin{subfigure}[b]{0.06\textwidth}
    \includegraphics[width = \textwidth]{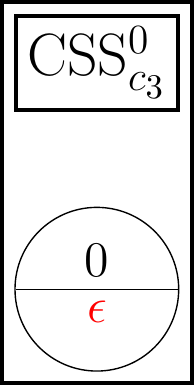}
    \label{fig:jupiter-css-podc16-c1423-css0}
  \end{subfigure}
  \hfil
  \begin{subfigure}[b]{0.07\textwidth}
    \includegraphics[width = \textwidth]{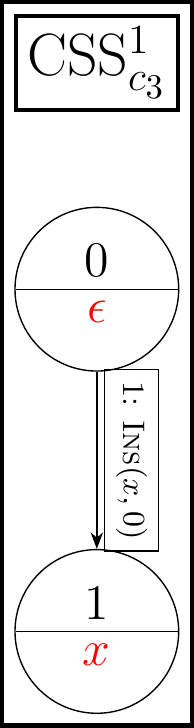}
    \label{fig:jupiter-css-podc16-c1423-css1}
  \end{subfigure}
  \hfil
  \begin{subfigure}[b]{0.17\textwidth}
    \includegraphics[width = \textwidth]{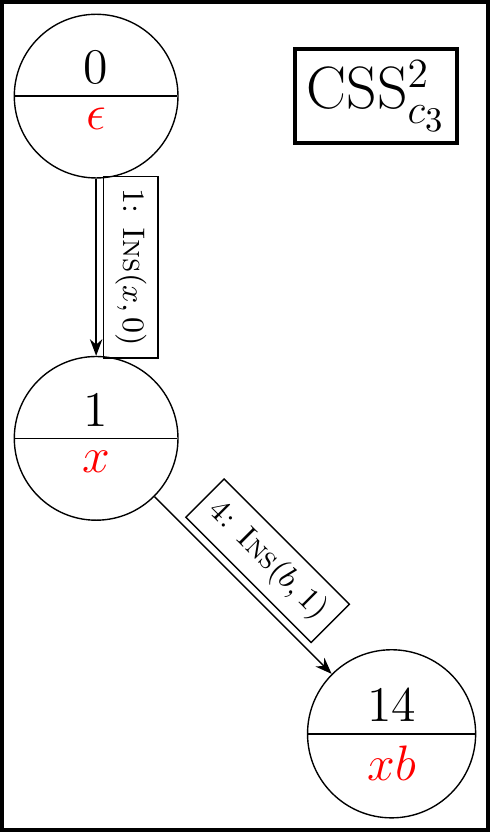}
    \label{fig:jupiter-css-podc-c1423-css2}
  \end{subfigure}
  \hfil
  \begin{subfigure}[b]{0.32\textwidth}
    \includegraphics[width = \textwidth]{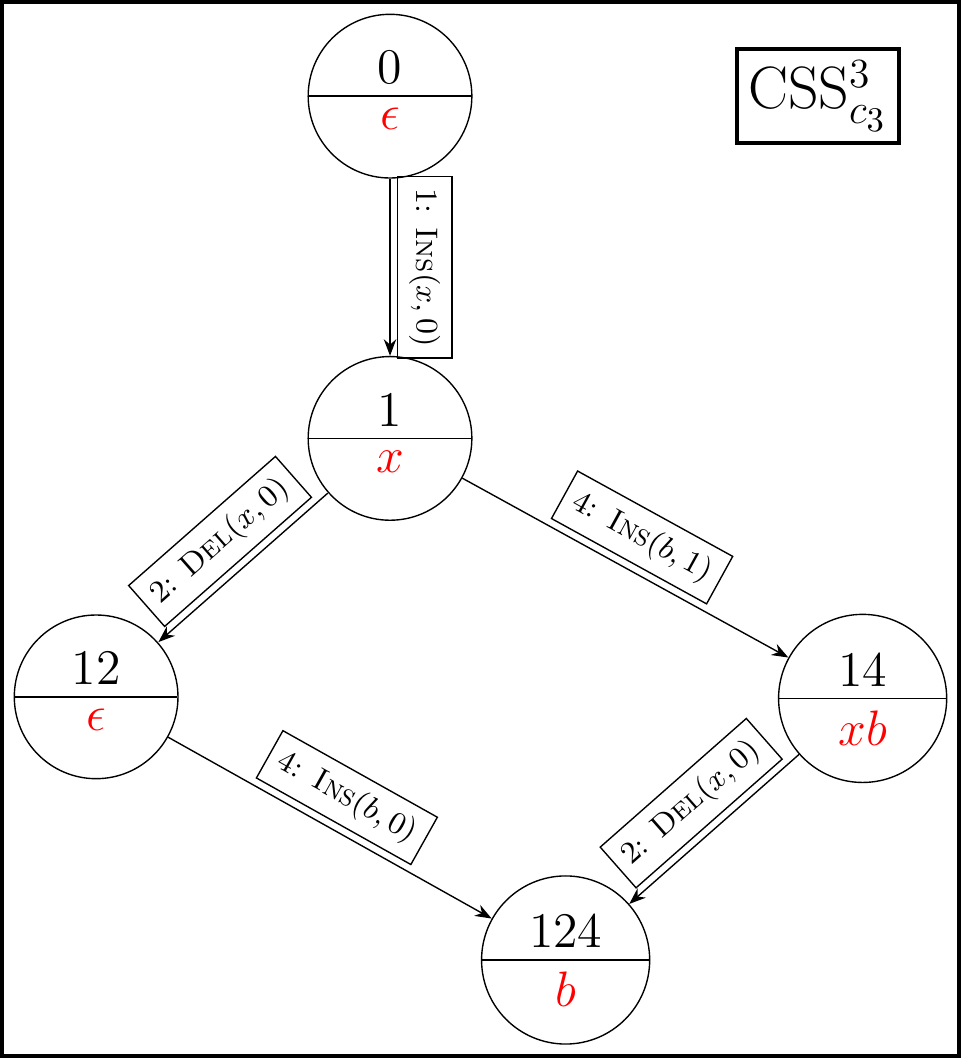}
    \label{fig:jupiter-css-podc16-c1423-css3}
  \end{subfigure}
  \hfil
  \begin{subfigure}[b]{0.35\textwidth}
    \includegraphics[width = \textwidth]{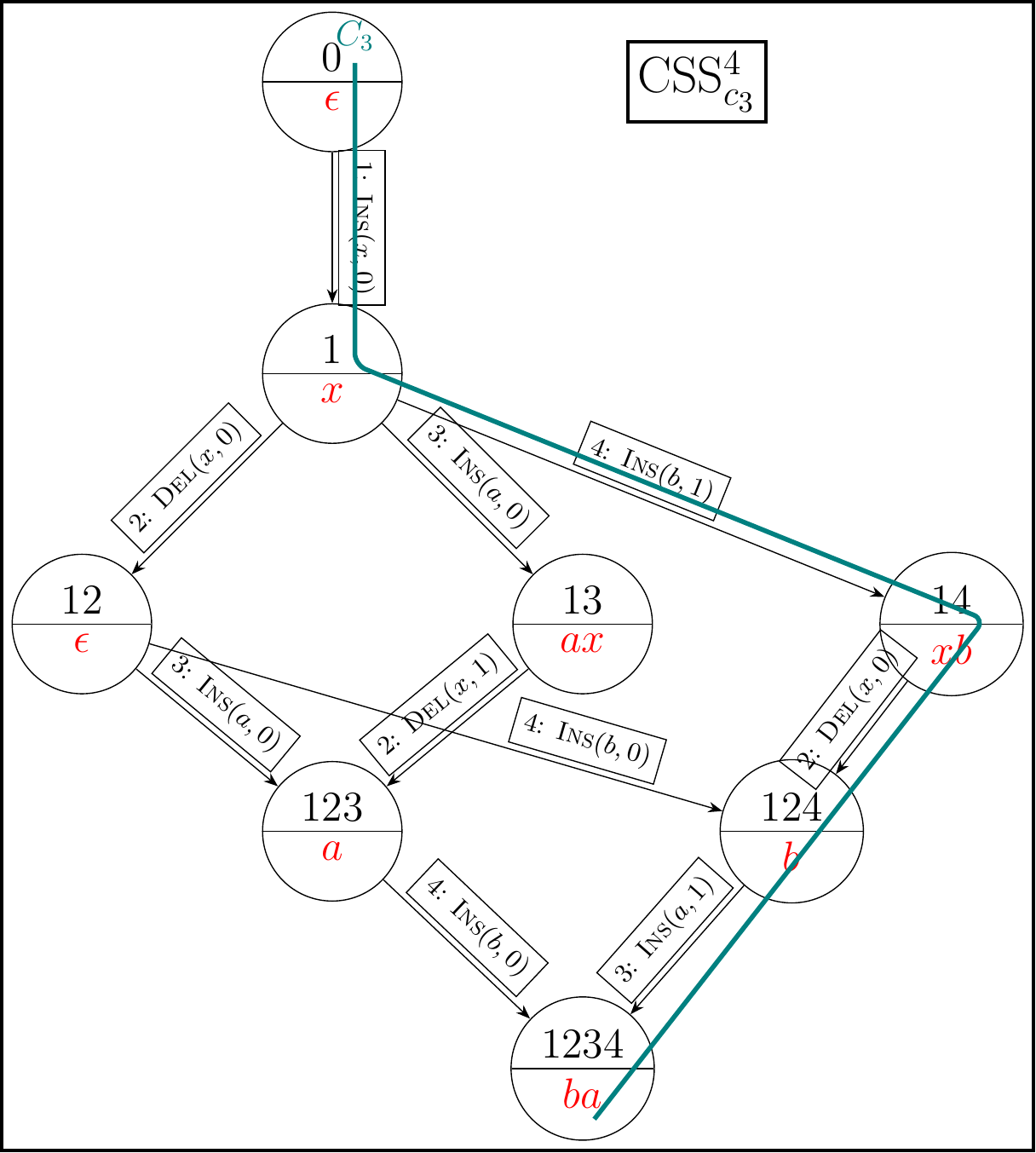}
    \label{fig:jupiter-css-podc16-c1423-css4}
  \end{subfigure}
  \caption{Illustration of client $c_3$ in \cjupiter{} under the schedule of Figure~\ref{fig:jupiter-schedule-podc16}.
  Its behavior (i.e., the sequence of state transitions) is indicated by the path in \cssck{3}{4}.
  (Please refer to Figure~\ref{fig:appendix-cjupiter-illustration} for the illustration of clients $c_1$ and $c_2$ and the server $s$.)}
  \label{fig:jupiter-css-podc16-c1423}
\end{figure}

To keep track of the construction of the $n$-ary ordered state spaces in \cjupiter,
for each state space, we introduce a superscript $k$ to refer to the one after the $k$-th step 
(i.e., after processing $k$ operations), counting from $0$.
For instance, the state space \cssc{i} (resp. \csss) after the $k$-th step
maintained by client $c_i$ (resp. the server $s$) is denoted by \cssck{i}{k} (resp. \csssk{k}).
This notational convention also applies to \jupiter{} (reviewed in Section~\ref{ss:jupiter-protocol}).

\begin{example}[Illustration of \cjupiter]   \label{ex:cjupiter}

Figure~\ref{fig:jupiter-css-podc16-c1423} illustrates client $c_3$ 
in \cjupiter{} under the schedule of Figure~\ref{fig:jupiter-schedule-podc16}.
For convenience, we denote, for instance, 
a vertex $v$ with $v.oids = \set{o_1, o_4}$ by $v_{14}$ and
an operation $o_3$ with $o_3.ctx = \set{o_1, o_2}$ by $\copinmath{o_3}{o_1, o_2}$.
We have also mixed the notations of operations of types $\opset$ and \Op{}
when no confusion arises.
We map various vertices and operations in this example to the ones 
(i.e., $u, u', v, v', op, op'$) used in the description of the \cjupiter{} protocol.

After receiving and applying $o_1 = \lins(x,0)$ of client $c_1$ from the server,
client $c_3$ generates $o_4 = \lins(b,1)$.
It applies $o_4$ locally, creates a new vertex $v_{14}$, and appends it to \cssck{3}{1}
via an edge from $v_1$ labeled with $\copinmath{o_4}{o_1}$.
Then, $\copinmath{o_4}{o_1}$ is propagated to the server.

Next, client $c_3$ receives $o_2 = \ldel(x,0)$ of client $c_1$ from the server.
The operation context of $o_2$ is $\set{o_1}$, matching the $oids$ of $v_{1}$ ($u$).
By \textsc{xForm}, $\copinmath{o_2}{o_1}$ ($op$) is transformed with $\copinmath{o_4}{o_1}$ ($op'$):
$OT\Big(\copinmath{o_2}{o_1} = \ldel(x,0), \copinmath{o_4}{o_1} = \lins(b,1)\Big)
  = \Big(\copinmath{o_2}{o_1,o_4} = \ldel(x,0), \copinmath{o_4}{o_1,o_2} = \lins(b,0)\Big)$.
As a result, $v_{124}$ ($v'$) is created and is linked to $v_{12}$ ($v$) and $v_{14}$ ($u'$) 
via the edges labeled with $\copinmath{o_4}{o_1,o_2}$ and $\copinmath{o_2}{o_1,o_4}$, respectively.
Because $o_2$ is unaware of $o_4$ at the server ($o_4.sctx = \emptyset$ now),
the edge from $v_{1}$ to $v_{12}$ is ordered before (to the left of) 
that from $v_{1}$ to $v_{14}$ in \cssck{3}{3}.

Finally, client $c_3$ receives $\copinmath{o_3}{o_1} = \lins(a,0)$ of client $c_2$ from the server.
The operation context of $o_3$ is $\set{o_1}$, matching the $oids$ of $v_1$ ($u$).
By $\textsc{xForm}$, $\copinmath{o_3}{o_1}$ will be transformed with the operation sequence
consisting of operations along the \emph{first} edges from $v_1$ to the final vertex $v_{124}$ of \cssck{3}{3}, 
namely $\copinmath{o_2}{o_1}$ from $v_{1}$ and $\copinmath{o_4}{o_1,o_2}$ from $v_{12}$.
Specifically, $\copinmath{o_3}{o_1}$ ($op$) is first transformed with $\copinmath{o_2}{o_1}$ ($op'$):
$OT\Big(\copinmath{o_3}{o_1} = \lins(a,0), \copinmath{o_2}{o_1} = \ldel(x,0)\Big)
  = \Big(\copinmath{o_3}{o_1,o_2} = \lins(a,0), \copinmath{o_2}{o_1,o_3} = \ldel(x,1)\Big)$.
Since $o_3$ is aware of $o_2$ but unaware of $o_4$ at the server,
the new edge from $v_{1}$ labeled with $\copinmath{o_3}{o_1}$ 
is placed before that with $\copinmath{o_4}{o_1}$ but after that with $\copinmath{o_2}{o_1}$.
Then, $\copinmath{o_3}{o_1,o_2}$ ($op$) is transformed with $\copinmath{o_4}{o_1, o_2}$ ($op'$),
yielding $v_{1234}$ and $\copinmath{o_3}{o_1,o_2,o_4}$.
Client $c_3$ applies $\copinmath{o_3}{o_1, o_2, o_4}$, obtaining the list content $ba$.
\end{example}


The choice of the ``first'' edges in OTs 
is necessary to establish equivalence between \cjupiter{} and \jupiter{},
particularly \emph{at the server side}.
First, the operation sequence along the first edges from a vertex of \csss{} at the server
admits a simple characterization.

\begin{lemma}[\cjupiter{}'s ``First'' Rule]  \label{lemma:cjupiter-first-rule}
  Let $OP = \seq{op_1, op_2, \ldots, op_m}$ ($op_i \in \Op{}$)
  be the operation sequence the server has currently processed in total order `$\prec_{s}$'.
  For any vertex $v$ in the current \csss{},
  the path along the \textbf{first} edges from $v$ to the final vertex of \csss{}
  consists of the operations of $OP \setminus v$ in total order `$\prec_{s}$'
  (may be empty if $v$ is the final vertex of \csss{}), where
  \[
    OP \setminus v = \Big\{op \in OP \mid op.oid \in \set{op_1.oid, op_2.oid, \cdots, op_m.oid} \setminus v.oids \Big\}.
  \]
\end{lemma}

\begin{example}[\cjupiter{}'s ``First'' Rule]	\label{ex:cjupiter-first-rule}
  Consider \csss{} at the server shown in Figure~\ref{fig:cjupiter-css-podc16-allinone}
  under the schedule of Figure~\ref{fig:jupiter-schedule-podc16};
  see Figure~\ref{fig:cjupiter-illustration-server} for its construction.
  Suppose that the server has processed all four operations.
  That is, we take $OP = \seq{o_1, o_2, o_3, o_4}$ in Lemma~\ref{lemma:cjupiter-first-rule} 
  (we mix operations of types $\opset{}$ and $\Op{}$).
  Then, the path along the first edges from vertex $v_{1}$ (resp. $v_{13}$) 
  consists of the operations $OP \setminus v_{1} = \set{o_2, o_3, o_4}$
  (resp. $OP \setminus v_{13} = \set{o_2, o_4}$) in total order `$\prec_{s}$'. 
\end{example}

Based on Lemma~\ref{lemma:cjupiter-first-rule}, 
the operation sequence with which an operation transforms \emph{at the server}
can be characterized as follows,
which is exactly the same with that for \jupiter{}~\cite{Xu:CSCW14}.

\begin{lemma}[\cjupiter{}'s OT Sequence] \label{lemma:cjupiter-ot-server}
  In \emph{\textsc{xForm}} of \emph{\cjupiter{}}, 
  the operation sequence $L$ (may be empty) with which an operation $op$ transforms \textbf{at the server}
  consists of the operations that are both totally ordered by `$\prec_{s}$' before 
  and concurrent by \emph{`$\parallel$'} with $op$.
  Furthermore, the operations in $L$ are totally ordered by `$\prec_{s}$'.
\end{lemma}

\begin{example}[\cjupiter{}'s OT Sequence]	\label{ex:cjupiter-ot-sequence}
  Consider the behavior of the server summarized in Figure~\ref{fig:cjupiter-css-podc16-allinone}
  under the schedule of Figure~\ref{fig:jupiter-schedule-podc16}.
  According to Lemma~\ref{lemma:cjupiter-first-rule}, 
  the operation sequence with which $op = o_4$ transforms 
  consists of operations $o_2$ (i.e., $\copinmath{o_2}{o_1}$) from vertex $v_1$
  and $o_3$ (i.e., $\copinmath{o_3}{o_1, o_2}$) from vertex $v_{12}$
  in total order `$\prec_{s}$',
  which are both totally ordered by `$\prec_{s}$' before and concurrent by `$\parallel$' with $o_4$.
\end{example}
\subsection{\cjupiter{} is Compact} \label{ss:cjupiter-compact}

Although ($n+1$) $n$-ary ordered state spaces are maintained by \cjupiter{} for a system with $n$ clients,
they are all the same.
That is, at a high level, \cjupiter{} maintains only a single $n$-ary ordered state space.

\begin{prop}[$n+1 \Rightarrow 1$] \label{prop:css-server-client}
  In \emph{\cjupiter{}}, the replicas that have processed the same set of operations 
  (in terms of their $oid$s) have the same $n$-ary ordered state space.
\end{prop}

Informally, this proposition holds because 
we have kept all ``by-product'' states/vertices of OTs in the $n$-ary ordered state spaces,
and each client is ``synchronized'' with the server.
Since all replicas will eventually process all operations,
the final $n$-ary ordered state spaces at all replicas are the same.
The construction order may differ replica by replica.

\begin{example}[\cjupiter{} is Compact]
  Figure~\ref{fig:cjupiter-css-podc16-allinone} shows the same final $n$-ary ordered state space 
  constructed by \cjupiter{} for each replica under the schedule of Figure~\ref{fig:jupiter-schedule-podc16}.
  (Figure~\ref{fig:appendix-cjupiter-illustration} shows the step-by-step construction for each replica.)
  Each replica behavior (i.e., the sequence of state transitions) corresponds to a path going through this state space.
  As illustrated, the server $s$ and client $c_1$ 
  go along the path $v_{0} \xrightarrow{o_1} v_{1} \xrightarrow{o_2} v_{12} \xrightarrow{o_3} v_{123} \xrightarrow{o_4} v_{1234}$,
  client $c_2$ goes along the path $v_{0} \xrightarrow{o_1} v_{1} \xrightarrow{o_3} v_{13} \xrightarrow{o_2} v_{123} \xrightarrow{o_4} v_{1234}$,
  and client $c_3$ goes along the path $v_{0} \xrightarrow{o_1} v_{1} \xrightarrow{o_4} v_{14} \xrightarrow{o_2} v_{124} \xrightarrow{o_3} v_{1234}$.
\end{example}

Together with the fact that the OT functions satisfy CP1, 
Proposition~\ref{prop:css-server-client} implies that

\begin{theorem}[$\cjupiter{} \models \cp$]   \label{thm:css-cp}
  \emph{\cjupiter{}} satisfies the convergence property $\cp{}$.
\end{theorem}
\section{\cjupiter{} is Equivalent to \jupiter}   \label{section:jupiter-cjupiter-equiv}

We now prove that \cjupiter{} is equivalent to \jupiter{}
(reviewed in Section~\ref{ss:jupiter-protocol})
from perspectives of both the server and clients.
Specifically, we prove that
the behaviors of the servers are the same (Section~\ref{ss:server-equiv}), 
and that the behaviors of each pair of corresponding clients are the same (Section~\ref{ss:client-equiv}).
Consequently, we have that

\begin{theorem}[Equivalence]    \label{thm:equiv}
  Under the same schedule, the behaviors (Section~\ref{ss:model}) of corresponding replicas 
  in \emph{\cjupiter{}} and \emph{\jupiter{}} are the same.
\end{theorem}

\subsection{Review of \jupiter{}}  \label{ss:jupiter-protocol}

We review the \jupiter{} protocol in~\cite{Xu:CSCW14},
a \emph{multi-client} description of \jupiter{} first proposed in~\cite{Nichols:UIST95}~\footnote{
  The Jupiter protocol in~\cite{Nichols:UIST95} uses 1D buffers,
  but does not explicitly describe the multi-client scenario.}.
Consider a client/server system with $n$ clients.
\jupiter{}~\cite{Xu:CSCW14} maintains $2n$ \emph{2D state spaces} 
(Appendix~\ref{appendix:ss-2d-state-space}),
each consisting of a \emph{local} dimension and a \emph{global} dimension.
Specifically, each client $c_i$ maintains a 2D state space, denoted \cscwc{i},
with the local dimension for operations generated by the client
and the global dimension by others.
The server maintains $n$ 2D state spaces, one for each client.
The state space for client $c_i$, denoted \cscws{i},
consists of the local dimension for operations from client $c_i$
and the global dimension from others.

\jupiter{} is similar to \cjupiter{} with two major differences:
First, in $\textsc{xForm}(op: \Op, d \in \set{\Local, \Remote})$ of \jupiter{}, 
the operation sequence with which $op$ transforms is determined by the parameter $d$,
indicating the local/global dimension described above
(instead of following the \emph{first} edges as in \cjupiter{}).
Second, in \jupiter{}, the server propagates the \emph{transformed} operation 
(instead of the original one it receives) to other clients.
As with \cjupiter{}, we describe \jupiter{} in three parts.
We omit the details that are in common with and have been explained in \cjupiter{};
see Appendix~\ref{appendix:ss-jupiter} for pseudocode.

{\bf Local Processing Part.}
When client $c_i$ receives an operation $o \in \opset$ from a user,
it applies $o$ locally, generates $op \in \Op$ for $o$,
saves $op$ along the local dimension at the end of its 2D state space \cscwc{i},
and sends $op$ to the server asynchronously.

{\bf Server Processing Part.}
When the server receives an operation $op \in \Op$ from client $c_i$, it
first transforms $op$ with an operation sequence along the global dimension 
in \cscws{i} to obtain $op'$ by calling $\textsc{xForm}(op, \Remote)$ (see below),
and applies $op'$ locally.
Then, for each $j \neq i$, it saves $op'$ at the end of \cscws{j} along the global dimension.
Finally, $op'$ (instead of $op$) is sent to other clients asynchronously.

{\bf Remote Processing Part.}
When client $c_i$ receives an operation $op \in \Op$ from the server,
it transforms $op$ with an operation sequence along the local dimension 
in its 2D state space \cscwc{i} to obtain $op'$ 
by calling $\textsc{xForm}(op, \Local)$ (see below),
and applies $op'$ locally.

{\bf OTs in \jupiter{}.}
In the procedure $\textsc{xForm}(op: \Op, d: \LR = \set{\Local, \Remote})$ of \jupiter{},
the operation sequence with which $op$ transforms is determined by an extra parameter $d$.
Specifically, it first locates the vertex $u$ whose $oids$ matches the operation context $op.ctx$ of $op$,
and then iteratively transforms $op$ with an operation sequence along the $d$ dimension from $u$ 
to the final vertex of this 2D state space.


\begin{example}[Illustration of \jupiter{}]   \label{ex:jupiter}

  \begin{figure}[t]
    \begin{sideways}
     \includegraphics[width = 1.00\textwidth]{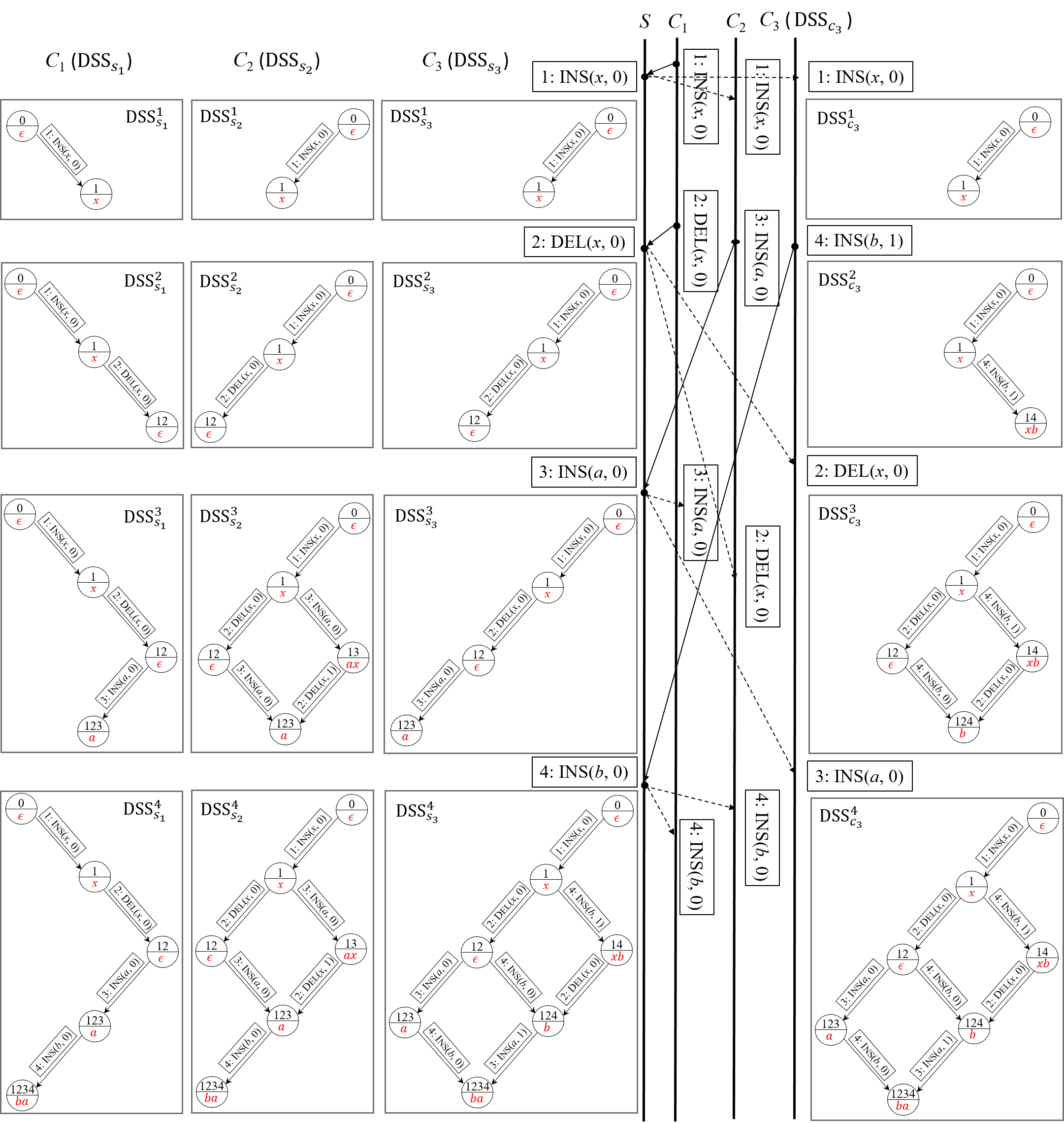}
    \end{sideways}
    \centering
    \caption{(Rotated) illustration of client $c_3$, as well as the server $s$, in \jupiter{}~\cite{Xu:CSCW14}
    	under the schedule of Figure~\ref{fig:jupiter-schedule-podc16}.
	(Please refer to Figure~\ref{fig:cjupiter-illustration} for details of clients $c_1$ and $c_2$.)}
    \label{fig:jupiter-illustration}
  \end{figure}

  Figure~\ref{fig:jupiter-illustration} illustrates client $c_3$, as well as the server $s$, in \jupiter{} 
  under the schedule of Figure~\ref{fig:jupiter-schedule-podc16}.
  The first three state transitions made by client $c_3$ in \jupiter{}
  due to the operation sequence consisting of $o_1$ from client $c_1$, $o_4$ generated by itself,
  and $o_2$ from client $c_1$ are the same with those in \cjupiter{};
  see \cssck{3}{1}, \cssck{3}{2}, and \cssck{3}{3} of Figure~\ref{fig:jupiter-css-podc16-c1423}
  and \cscwck{3}{1}, \cscwck{3}{2}, and \cscwck{3}{3} of Figure~\ref{fig:jupiter-illustration}.

  We now elaborate on the fourth state transition of client $c_3$ in \jupiter{}. 
  First, client $c_2$ propagates its operation $\copinmath{o_3}{o_1} = \lins(a,0)$ to the server $s$.
  At the server, $\copinmath{o_3}{o_1}$ is transformed with $\copinmath{o_2}{o_1} = \ldel(x,0)$ in \cscwsk{2}{3}, 
  obtaining $\copinmath{o_3}{o_1, o_2} = \lins(a,0)$.
  In addition to being stored in \cscwsk{1}{3} and \cscwsk{3}{3}, 
  the transformed operation $\copinmath{o_3}{o_1, o_2}$ is then redirected by the server to clients $c_1$ and $c_3$.
  At client $c_3$, the operation context of $\copinmath{o_3}{o_1, o_2}$ (i.e., $\set{o_1, o_2}$)
  matches the $oids$ of $v_{12}$ ($u$) in \cscwck{3}{4}.
  By \textsc{xForm}, $\copinmath{o_3}{o_1, o_2}$ ($op$) is transformed with $\copinmath{o_4}{o_1, o_2}$ ($op'$),
  yielding $v_{1234}$ and $\copinmath{o_3}{o_1,o_2,o_4}$.
  Finally, client $c_3$ applies $\copinmath{o_3}{o_1, o_2, o_4}$, obtaining the list content $ba$.

  We highlight three differences between \cjupiter{} and \jupiter{},
  by comparing the behaviors of client $c_3$ in this example and Example~\ref{ex:cjupiter}.
  First, the fourth operation the server $s$ redirects to client $c_3$ 
  is the transformed operation $\copinmath{o_3}{o_1, o_2} = \lins(a,0)$,
  instead of the original one 
  $\copinmath{o_3}{o_1} = \lins(a,0)$~\footnote{Although they happen to have the same signature $\lins(a,0)$, 
    they have different operation contexts.} generated by client $c_2$.
  Second, each vertex in the $n$-ary ordered state space of \cjupiter{} 
  (such as \cssck{3}{4} of Figure~\ref{fig:jupiter-css-podc16-c1423}) is not restricted to have only two child vertices,
  while \jupiter{} does.
  Third, because the transformed operations are propagated by the server,
  \jupiter{} is slightly optimized in implementation \emph{at clients} by eliminating redundant OTs.
  For example, in \cssck{3}{4} of Figure~\ref{fig:jupiter-css-podc16-c1423},
  the original operation $\copinmath{o_3}{o_1}$ of client $c_2$ redirected by the server 
  should be first transformed with $\copinmath{o_2}{o_1}$ to obtain $\copinmath{o_3}{o_1, o_2}$.
  In \jupiter, however, such a transformation which has been done at the server (i.e., in \cscwsk{2}{3})
  is not necessary at client $c_3$ (i.e., in \cscwck{3}{4}).
\end{example}



\subsection{The Servers Established Equivalent} \label{ss:server-equiv}

As shown in~\cite{Xu:CSCW14} (see the ``Jupiter'' section and Definition~8 of~\cite{Xu:CSCW14}),
the operation sequence with which an incoming operation transforms \emph{at the server}
in \textsc{xForm} of \jupiter{} can be characterized exactly as in \textsc{xForm} of \cjupiter{}
(Lemma~\ref{lemma:cjupiter-ot-server}).
By mathematical induction on the operation sequence the server processes,
we can prove that the state spaces of \jupiter{} and \cjupiter{} at the server are essentially the same.
Formally, the $n$-ary ordered state space \csss{} of \cjupiter{}
equals the union~\footnote{The union is taken on state spaces 
which are (directed) graphs as sets of vertices and edges.
The order of edges of $n$-ary ordered state spaces should be respected when \cscws{i}'s are unioned to obtain \csss{}.}
of all 2D state spaces \cscws{i}
maintained at the server for each client $c_i$ in \jupiter{}.
For example, \csss{} of Figure~\ref{fig:cjupiter-css-podc16-allinone} is the union of
the three \cscws{i}'s of Figure~\ref{fig:jupiter-illustration}.
More specifically, we have
\begin{prop}[$n \leftrightarrow 1$]  \label{prop:server-equiv}
  Suppose that under the same schedule, the server has processed a sequence of $m$ operations,
  denoted $O = \seq{op_1, op_2, \ldots, op_m}$ ($op_i \in \emph{\Op}$), in total order `$\prec_{s}$'.
  We have that
  \begin{equation*} \label{eq:server-equiv}
    \cssskinmath{k} = \bigcup_{i=1}^{i=k} \cscwskinmath{c(op_{i})}{i} =
      \bigcup_{c_i \in c(O)} \bigcup_{j=1}^{j=k} \cscwskinmath{c_i}{j}, \;\; 1 \leq k \leq m, \tag{$\ast$}
  \end{equation*}
  where $\text{c}(op_i)$ denotes the client that generates the operation $op_i$ (more specifically, $op_i.o$)
  and $c(O) = \set{c(op_1), c(op_2), \ldots, c(op_m)}$.
\end{prop}

The equivalence of servers are thus established.
\begin{theorem}[Equivalence of Servers] \label{thm:server-equiv}
  Under the same schedule, the behaviors 
  (i.e., the sequence of (list) state transitions, defined in Section~\ref{ss:model}) 
  of the servers in \emph{\cjupiter{}} and \emph{\jupiter{}} are the same.
\end{theorem}

\subsection{The Clients Established Equivalent} \label{ss:client-equiv}

As discussed in Example~\ref{ex:jupiter},
\jupiter{} is slightly optimized in implementation \emph{at clients} by eliminating redundant OTs.
Formally, by mathematical induction on the operation sequence client $c_i$ processes,
we can prove that \cscwck{i}{k} of \jupiter{} is a part (i.e., subgraph) of \cssck{i}{k} of \cjupiter{}.
The equivalence of clients follows since the final transformed operations (for an original one)
executed at $c_i$ in \jupiter{} and \cjupiter{} are the same, 
regardless of the optimization adopted by \jupiter{} at clients.

\begin{prop}[$1 \leftrightarrow 1$]  \label{prop:client-equiv}
  Under the same schedule, we have that
  \begin{equation*}  \label{eq:client-equiv}
    \cscwckinmath{i}{k} \subseteq \cssckinmath{i}{k}, \quad 1 \leq i \leq n, \,k \geq 1. \tag{$\star$}
  \end{equation*}
\end{prop}

\begin{theorem}[Equivalence of Clients]  \label{thm:client-equiv}
  Under the same schedule, the behaviors (Section~\ref{ss:model}) of each pair of corresponding clients
  in \emph{\cjupiter{}} and \emph{\jupiter{}} are the same.
\end{theorem}



\section{\cjupiter{} Satisfies the Weak List Specification}  \label{section:cjupiter-weak-spec}

The following theorem, together with Theorem~\ref{thm:equiv},
solves the conjecture of Attiya et al.~\cite{Attiya:PODC16}.

\begin{theorem}[$\cjupiter{} \models \wlspec{}$]   \label{thm:wl-spec}
  \emph{\cjupiter{}} satisfies the weak list specification $\wlspec{}$.
\end{theorem}

\begin{proof}
  For each execution $\alpha$ of \cjupiter{},
  we construct an abstract execution $A = (H, \vis{})$ with $\vis{} = \;\hbrel{\alpha}{}{}$
  (Section~\ref{ss:model}).
  We then prove the conditions of $\wlspec{}$ (Definition~\ref{def:wl-spec}) in the order
  1(c), 1(a), 1(b), and 2.
  
  Condition 1(c) follows from the local processing of \cjupiter{}.
  Condition 1(a) holds due to the FIFO communication 
  and the property of OTs that when transformed in \cjupiter{}, 
  the type and effect of an $\ins{a}{p}$ (resp. a $\del{a}{p}$)
  remains unchanged (with a trivial exception of being transformed to be \Nop{}),
  namely to insert (resp. delete) the element $a$ (possibly at a different position than $p$).
  
  To show that $A = (H, \vis{})$ belongs to $\wlspec{}$,
  we define the list order relation \lo{} in Definition~\ref{def:lo} below,
  and then prove that \lo{} satisfies conditions 1(b) and 2 of Definition~\ref{def:wl-spec}.
\end{proof}

\begin{definition}[List Order `$\lo{}$']  \label{def:lo}
  Let $\alpha$ be an execution.
  For $a, b \in \elems{A}$, $\lorel{a}{b}$ if and only if there exists an event $e \in \alpha$
  with returned list $w$ such that $a$ precedes $b$ in $w$.
\end{definition}

By definition,
\itshape 1) \upshape
  \lo{} is \emph{transitive} and \emph{total}
    on $\set{a \mid a \in w}$ for all events $e = \doe{}(o, w) \in H$; \emph{and}
\itshape 2) \upshape
  \lo{} satisfies 1(b) of Definition~\ref{def:wl-spec}.
The \emph{irreflexivity} of \lo{} can be rephrased 
in terms of the pairwise state compatibility property.

\begin{definition}[State Compatibility]  \label{def:state-compat}
  Two list states $w_1$ and $w_2$ are \emph{compatible},
  if and only if for any two common elements $a$ and $b$ of $w_1$ and $w_2$,
  their relative orderings are the same in $w_1$ and $w_2$.
\end{definition}

\begin{lemma}[Irreflexivity]  \label{lemma:irreflexivity}
  Let $\alpha$ be an execution and $A = (H, \vis{})$ the abstract execution 
  constructed from $\alpha$ as described in the proof of Theorem~\ref{thm:wl-spec}.
  The list order \emph{\lo{}} based on $\alpha$ is irreflexive
  if and only if the list states (i.e., returned lists) in $A$ are pairwise compatible.
\end{lemma}

The proof relies on the following lemma about paths in $n$-ary ordered state spaces.
\begin{lemma}[Simple Path]  \label{lemma:simple-path}
  Let $P_{v_1 \leadsto v_2}$ be a path 
  from vertex $v_1$ to vertex $v_2$ in an $n$-ary ordered state space.
  Then, there are no duplicate operations (in terms of their $oid$s)
  along the path $P_{v_1 \leadsto v_2}$.
  We call such a path a simple path.
\end{lemma}


Therefore, it remains to prove that all list states in an execution of \cjupiter{} are pairwise compatible,
which concludes the proof of Theorem~\ref{thm:wl-spec}.
By Proposition~\ref{prop:css-server-client},
we can focus on the state space \csss{} at the server.
We first prove several properties about vertex pairs and paths of \csss{},
which serve as building blocks for the proof of the main result (Theorem~\ref{thm:pcp}).

By mathematical induction on the operation sequence processed
in the total order $\prec_{s}$ at the server and by contradiction (in the inductive step),
we can show that

\begin{lemma}[LCA]   \label{lemma:lca}
  In \emph{\cjupiter{}}, each pair of vertices in the $n$-ary ordered state space \csss{}
  (as a rooted directed acyclic graph)
  has a unique LCA (Lowest Common Ancestor).~\footnote{The LCAs of
    two vertices $v_1$ and $v_2$ in a rooted directed acyclic graph
  is a set of vertices $V$ such that
  \itshape 1) \upshape Each vertice in $V$ has both $v_1$ and $v_2$ as descendants;
  \itshape 2) \upshape In $V$, no vertice is an ancestor of another.
  The uniqueness further requires $|V| = 1$.}
\end{lemma}

In the following, we are concerned with the paths to a pair of vertices from their LCA.

\begin{lemma}[Disjoint Paths]    \label{lemma:disjoint-paths}
  Let $v_0$ be the unique LCA of a pair of vertices $v_1$ and $v_2$
  in the $n$-ary ordered state space \csss{}, denoted $v_0 = \text{LCA}(v_1, v_2)$.
  Then, the set of operations $O_{v_0 \leadsto v_1}$ 
  along a simple path $P_{v_0 \leadsto v_1}$ 
  is disjoint in terms of the operation oids from
  the set of operations $O_{v_0 \leadsto v_2}$ along a simple path $P_{v_0 \leadsto v_2}$.
\end{lemma}


The next lemma gives a sufficient condition for two states (vertices) being compatible
in terms of disjoint simple paths to them from a common vertex.

\begin{lemma}[Compatible Paths]  \label{lemma:compatible-paths}
  Let \pathplain{v_0}{v_1} and \pathplain{v_0}{v_2} 
  be two paths from vertex $v_0$ to vertices $v_1$ and $v_2$, 
  respectively in the $n$-ary ordered state space \csss{}.
  If they are disjoint simple paths, 
  then the list states of $v_1$ and $v_2$ are compatible.
\end{lemma}


The desired pairwise state compatibility property follows,
when we take the common vertex $v_0$ in Lemma~\ref{lemma:compatible-paths} 
as the LCA of the two vertices $v_1$ and $v_2$ under consideration.

\begin{theorem}[Pairwise State Compatibility]   \label{thm:pcp}
  Every pair of list states in the state space \csss{} are compatible.
\end{theorem}

\begin{proof}
  Consider vertices $v_1$ and $v_2$ in \csss{}.
  \itshape 1) \upshape
    By Lemma~\ref{lemma:lca}, they have a unique LCA, denoted $v_0$;
  \itshape 2) \upshape
    By Lemma~\ref{lemma:simple-path}, \pathplain{v_0}{v_1} and \pathplain{v_0}{v_2} are simple paths;
  \itshape 3) \upshape
    By Lemma~\ref{lemma:disjoint-paths}, \pathplain{v_0}{v_1} and \pathplain{v_0}{v_2} are disjoint; and
  \itshape 4) \upshape
    By Lemma~\ref{lemma:compatible-paths}, the list states of $v_1$ and $v_2$ are compatible.
\end{proof}


\section{Related Work}  \label{section:related-work}

Convergence is the main property for implementing
a highly-available replicated list object~\cite{Ellis:SIGMOD89, Xu:CSCW14}.
Since 1989~\cite{Ellis:SIGMOD89}, 
a number of OT~\cite{Ellis:SIGMOD89}-based protocols have been proposed.
These protocols can be classified according to
whether they rely on a total order on operations~\cite{Xu:CSCW14}.
Various protocols like \jupiter~\cite{Nichols:UIST95, Xu:CSCW14} establish a total order
via a central server, a sequencer, or a distributed timestamping scheme
\cite{Wave, Vidot:CSCW00, Shen:CSCW02, Li:ICPADS04, Sun:TPDS09}.
By contrast, protocols like adOPTed~\cite{Ressel:CSCW96}
rely only on a partial (causal) order 
on operations~\cite{Ellis:SIGMOD89, Prakash:TOCHI94, Sun:TOCHI98, Sun:CSCW98, Sun:TOCHI02}.

In 2016, Attiya et al.~\cite{Attiya:PODC16} propose
the strong/weak list specification of a replicated list object.
They prove that the existing 
CRDT (Conflict-free Replicated Data Types)~\cite{Shapiro:SSS11}-based RGA protocol~\cite{Roh:JPDC11} 
satisfies the strong list specification,
and \emph{conjecture} that the well-known OT-based Jupiter protocol~\cite{Nichols:UIST95, Xu:CSCW14} 
satisfies the weak list specification.

The OT-based protocols typically use data structures like 1D buffer~\cite{Shen:CSCW02},
2D state space~\cite{Nichols:UIST95, Xu:CSCW14},
or $N$-dimensional interaction model~\cite{Ressel:CSCW96}
to keep track of OTs or choose correct OTs to perform.
As a generalization of 2D state space,
our $n$-ary ordered state space is similar to the $N$-dimensional interaction model.
However, they are proposed for different system models.
In an $n$-ary ordered state space, edges from the same vertex are \emph{ordered},
utilizing the existence of a total order on operations.
By contrast, the $N$-dimensional interaction model relies only on a partial order on operations.
Consequently, the simple characterization of OTs in \textsc{xForm} of \cjupiter{} 
does not apply in the $N$-dimensional interaction model.

\section{Conclusion and Future Work}    \label{section:conclusion}

We prove that the \jupiter{} protocol~\cite{Nichols:UIST95, Xu:CSCW14} 
satisfies the weak list specification~\cite{Attiya:PODC16},
thus solving the conjecture recently proposed by Attiya et al.~\cite{Attiya:PODC16}.
To this end, we have designed \cjupiter{}
based on a novel data structure called $n$-ary ordered state space.
In the future, we will explore how to algebraically manipulate and
reason about $n$-ary ordered state spaces.


\bibliography{jupiter-opodis18}

\clearpage
\appendix
\setcounter{footnote}{0}
\setcounter{figure}{0}
\renewcommand{\thefigure}{\thesection.\arabic{figure}}
\setcounter{algorithm}{0}
\renewcommand{\thealgorithm}{\thesection.\arabic{algorithm}}
\setlength{\belowcaptionskip}{0pt}

\section{The OT System}	\label{appendix:ot}

According to Section~\ref{ss:model}, 
we represent the replica state in an OT system (including both \jupiter{} and \cjupiter{})
as a sequence of operations $\seq{o_1, o_2, \cdots, o_m}$ (where, $o_i \in \opset$).

The function 
\[
  \textsc{Apply}: \Sigma \times \opset \to \Sigma \times \Val{}
\]
applies an operation $o$ to a state $\sigma$,
returning a new state $\sigma \circ o$ 
and the list content produced by performing $\sigma \circ o$ on the initial list.

\begin{figure*}[b]
  \centering
  \input{appendix/algs/list-ot}
  \caption{The OT functions satisfying CP1 for a replicated list object~\cite{Ellis:SIGMOD89, Imine:TCS06}.
    The parameter ``$pr$'' means ``priority'' which helps to resolve the conflicts
    when two concurrent \lins{} operations are intended to insert elements at the same position.
    In implementations, it is often to take the unique ids of replicas as priorities.
    The elements to be deleted in \ldel{} operations are irrelevant and are thus represented by `$\_$'s.
    \textsc{NOP} means ``do nothing''.
    Since we assume that all inserted elements are unique,
    the case of $OT\big(\lins(a_1, p_1, pr_1), \lins(a_2, p_2, pr_2)\big) = \textsc{NOP}$
    with $p_1 = p_2 \land a_1 = a_2$ in~\cite{Imine:TCS06} will never apply.}
  \label{fig:list-ot}
\end{figure*}


Figure~\ref{fig:list-ot} shows the OT functions satisfying CP1 
for a replicated list object~\cite{Ellis:SIGMOD89, Imine:TCS06}.
Operations \lins{} and \ldel{} have been extended with an extra parameter $pr$ for ``priority''~\cite{Imine:TCS06}.
It helps to resolve the conflicts when two concurrent \lins{} operations are intended to 
insert different elements at the same position.
We assume that the operations generated by the replica with a smaller identifier have a higher priority.
When a conflict occurs, 
the insertion position of the \lins{} operation with a higher priority will be shifted.

We highlight one property of OTs that when transformed in both \jupiter{} and \cjupiter{}, 
the type and effect of an insertion (resp. a deletion) $\ins{a}{p}$ (resp. $\del{a}{p}$) 
remains unchanged (with a trivial exception of being transformed to be \Nop{}), 
namely to insert (resp. delete) the element $a$ (possibly at a different position than $p$).

Figure~\ref{fig:xform-ot} illustrates an OT of two operations $op, op' \in \Op$
in both the $n$-ary ordered state space of \cjupiter{} 
and the 2D state space of \jupiter{}:
\[
  (\opot{op}{op'}, \opot{op'}{op}) = OT(op, op').
\]

Algorithm~\ref{alg:constants} lists the constants used in \jupiter{} and/or \cjupiter{}.

\begin{algorithm}
  \caption{Constants.}
  \label{alg:constants}
  \begin{algorithmic}[1]
    \LineComment{for both \jupiter{} and \cjupiter{}}
    \State $\SID{} = 0$
    \State $\CID{} = \set{1 \cdots n}$
    \State $\RID{} = \set{0 \cdots n}$
    \State $\SEQ{} = \mathbb{N}_{0}$

    \Statex
    \State \Enum{} \LR{} $\set{\Local = 0, \Remote = 1}$  \Comment{for \jupiter}
    \State \Enum{} \Ordering{} $\set{\Left = -1, \Right = 1}$\Comment{for \cjupiter}
  \end{algorithmic}
\end{algorithm}


\clearpage
\setcounter{figure}{0}
\section{The \cjupiter{} Protocol}	\label{section:appendix-cjupiter}

\subsection{Data Structure: $n$-ary Ordered State Space}	\label{section:appendix-css}

\begin{algorithm}
  \caption{Operation in \cjupiter{}.}
  \label{alg:css-op}
  \begin{algorithmic}[1]
    \CLASS{Op} 
      \State \VAR{o}\, : $\opset$
      \State \VAR{oid} : $\CID{} \times \SEQ{}$
      \State \VAR{ctx} : $2^{\CID{} \,\times\, \SEQ{}} = \emptyset$
      \State \VAR{sctx} : $2^{\CID{} \,\times\, \SEQ{}} = \emptyset$

      \Statex
      \Procedure{Compare}{$op : \Op, op' : \Op, r : \RID$} : \Ordering{}
	\If{$op.oid \in op'.sctx$}
	  \State \Return \Left{}	\Comment{$op \prec_{s} op'$}
	\ElsIf{$op'.oid \in op.sctx$}
	  \State \Return \Right{} 	\Comment{$op' \prec_{s} op$}

	  \LineComment{Here, $r$ must be a client replica, i.e., $r \in \CID{}$}
	  \ElsIf{$op.oid.cid \neq r$}	\Comment{$op$ is redirected by the server to client $r$}
	  \State \Return \Left{} 	\Comment{$op \prec_{s} op'$}
	\Else				\Comment{$op.oid.cid = r$. It must be the case that $op'.oid.cid \neq r$.}
	  \State \Return \Right{}	\Comment{$op' \prec_{s} op$}
	\EndIf
      \EndProcedure

      \Statex
      \Procedure{OT}{$op : \Op, op' : \Op$} : ($\Op, \Op$)
	\State $(o, o') \gets \Call{OT}{op.o, op.o'}$ \Comment{call OT on $\opset$}

	\hStatex
	\State $\Op\; \opot{op}{op'} = \text{ new } \Op(o, op.oid, op.ctx \cup \set{op'.oid}, op.sctx)$
	\State $\Op\; \opot{op'}{op} = \text{ new } \Op(o', op'.oid, op'.ctx \cup \set{op.oid}, op'.sctx)$

	\hStatex
	\State \Return $(\opot{op}{op'}, \opot{op'}{op})$
      \EndProcedure
    \EndClass{Op}
  \end{algorithmic}
\end{algorithm}

\begin{algorithm}
  \caption{Vertex in the $n$-ary ordered state space.}
  \label{alg:css-vertex}
  \begin{algorithmic}[1]
    \CLASS{Vertex} 
      \State \VAR{oids} : $2^{\CID{} \,\times\, \SEQ{}} = \emptyset$
      \State \VAR{edges} : $\SortedSet\langle\Edge\rangle = \emptyset$

      \Statex
      \Procedure{firstEdge}{$r : \RID$} : \Edge
	\State \Return the first edge according to $\Edge.\Call{Compare}{e : \Edge, e' : \Edge, r : \RID}$
      \EndProcedure


    \EndClass{Vertex}
  \end{algorithmic}
\end{algorithm}

\begin{algorithm}
  \caption{Edge in the $n$-ary ordered state space.}
  \label{alg:css-edge}
  \begin{algorithmic}[1]
    \CLASS{Edge}
      \State \VAR{op} : $\Op = \Null{}$
      \State \VAR{v} : $\Vertex = \Null{}$

      \Statex
      \Procedure{Compare}{$e : \Edge, e' : \Edge, r : \RID$} : \Ordering{}
	\State \Return \Call{Compare}{$e.op, e'.op, r$}
      \EndProcedure
    \EndClass{Edge}
  \end{algorithmic}
\end{algorithm}

\begin{algorithm}
  \caption{The $n$-ary ordered state space.}
  \label{alg:css-n-ary-state-space}
  \begin{algorithmic}[1]
    \CLASS{CStateSpace} 
      \State \VAR{cur} : \Vertex = new \Vertex()
      \State \VAR{r} : \RID




      \Statex
      \Procedure{xForm}{$op : \Op$} : \Op
      \State \Vertex{} $u \gets \Call{Locate}{op}$ \label{line:css-locate}
	\State \Vertex{} $v \gets$ new $\Vertex(u.oids \cup \set{op.oid}, \emptyset)$ \label{line:css-v}

	\hStatex
	\While{$u \neq cur$} \label{line:css-while-begin} \Comment{See Figure~\ref{fig:xform-ot}}
	  \State $\Edge\; e' \gets u.\Call{firstEdge}{r}$ 	\label{line:css-first-edge}
	  \State $\Vertex\; u' \gets e'.v$		 	\label{line:css-first-child}
	  \State $\Op\; op' \gets e'.op$ 			\label{line:css-first-op}

	  \hStatex
	  \State $(\opot{op}{op'}, \opot{op'}{op}) \gets \Call{OT}{op, op'}$  \label{line:css-ot}

	  \hStatex
	  \State \Vertex{} $v' \gets$ new $\Vertex(v.oids \cup \set{op'.oid}, \!\emptyset)$\label{line:css-v'}
	  \State $\Call{Link}{v, v', \opot{op'}{op}}$  \label{line:css-link-v-v'}
	  \State $\Call{Link}{u, v, op}$ \label{line:css-link-u-v}

	  \hStatex
	  \State $u \gets u'$ \label{line:css-update-begin}
	  \State $v \gets v'$
	  \State $op \gets \opot{op}{op'}$ \label{line:css-update-end}
	\EndWhile \label{line:css-while-end}

	\hStatex
	\State $\Call{Link}{u, v, op}$
        \State $cur \gets v$
	\State \Return $op$ \label{line:css-return}
      \EndProcedure

      \Statex
      \Procedure{Locate}{$op : \Op$} : \Vertex{}
	\State \Return \Vertex{} $v$ with $v.oids = op.ctx$
      \EndProcedure

      \Statex
      \Procedure{Link}{$u : \Vertex, v : \Vertex, op : \Op$}
	\State \Edge{} $e \gets$ new $\Edge(op, v)$
	\State $u.edges.\Call{add}{e}$
      \EndProcedure
    \EndClass{CStateSpace}
  \end{algorithmic}
\end{algorithm}

\clearpage
\subsection{The \cjupiter{} Protocol}	\label{section:appendix-cjupiter-protocol}
\begin{algorithm}
  \caption{Client in \cjupiter{}.}
  \label{alg:css-client}
  \begin{algorithmic}[1]
    \CLASS{Client} 
      \State \VAR{cid} : \CID{}
      \State \VAR{seq} : \SEQ{} = 0
      \State \VAR{state} : $\Sigma = \emptyseq$  \Comment{a sequence of $o \in \opset$}
      \State \VAR{S} : \CStateSpace = new \CStateSpace($cid$)

      \Statex
      \Procedure{Do}{$o : \opset$} : \Val{}  \Comment{Local Processing}
	\State $(state, val) \gets$ \Call{Apply}{$state, o$}

	\hStatex
	\State $seq \gets seq + 1$
	\State \Op{} $op \gets$ new $\Op(o, (cid, seq), S.cur.oids, \emptyset)$

	\hStatex
	\State \Vertex{} $v \gets$ new $\Vertex(S.cur.oids \cup \set{op.oid}, \emptyset)$
	\State $\Call{Link}{S.cur, v, op}$
	\State $S.cur \gets v$

	\hStatex
	\State \Call{Send}{$\SID{}, op$}  \Comment{send $op$ to the server}

	\State \Return $val$
      \EndProcedure

      \Statex
      \Procedure{Receive}{$op : \Op$}  \Comment{Remote Processing}
	\State $\Op\; op' \gets S.\Call{xForm}{op}$
	\State $state \gets state \circ op'.o$
      \EndProcedure
    \EndClass{Client}
  \end{algorithmic}
\end{algorithm}
\begin{algorithm}
  \caption{Server in \cjupiter{}.}
  \label{alg:css-server}
  \begin{algorithmic}[1]
    \CLASS{Server} 
      \State \VAR{state} : $\Sigma = \emptyseq$ \Comment{a sequence of $o \in \opset$}
      \State \VAR{soids} : $2^{\CID \times \SEQ} = \emptyset$
      \State \VAR{S} : \CStateSpace = new \CStateSpace(\SID)

      \Statex
      \Procedure{Receive}{$op : \Op$}  \Comment{Server Processing}
	\State $op.sctx \gets soids$
	\State $soids \gets soids \cup \set{op.oid}$

	\hStatex
	\State \Op{} $op' \gets S.\Call{xForm}{op}$
	\State $state \gets state \circ op'.o$

	\hStatex
	\ForAll{$c \in \CID \setminus \set{op.oid.cid}$}
	  \State \Call{Send}{$c, op$}	\Comment{send $op$ (not $op'$) to client $c$}
	\EndFor
      \EndProcedure
    \EndClass{Server}
  \end{algorithmic}
\end{algorithm}


\begin{figure}[t]
  \centering
  \begin{subfigure}[b]{1.00\textwidth}
    \begin{subfigure}[b]{0.06\textwidth}
      \includegraphics[width = \textwidth]{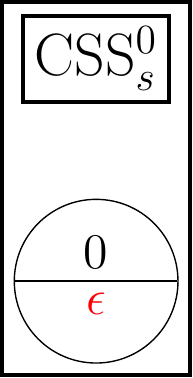}
    \end{subfigure}
    \hfil
    \begin{subfigure}[b]{0.07\textwidth}
      \includegraphics[width = \textwidth]{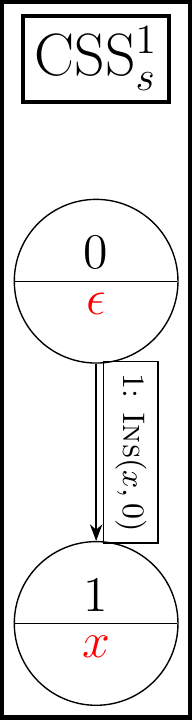}
    \end{subfigure}
    \hfil
    \begin{subfigure}[b]{0.17\textwidth}
      \includegraphics[width = \textwidth]{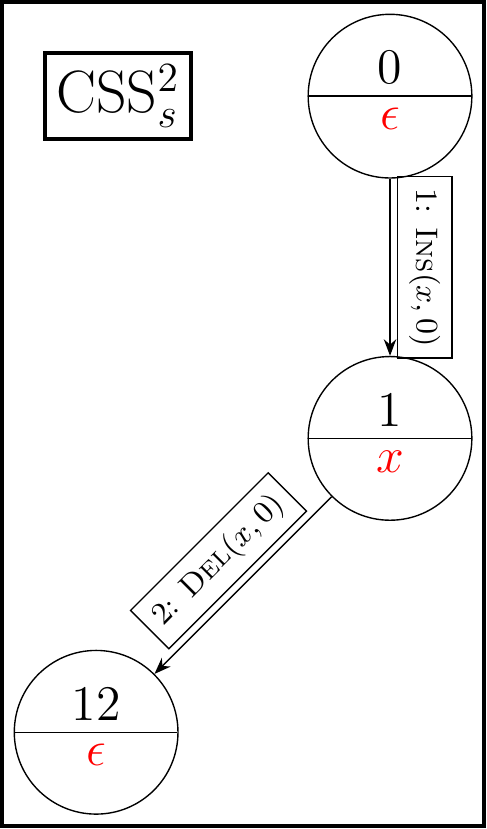}
    \end{subfigure}
    \hfil
    \begin{subfigure}[b]{0.28\textwidth}
      \includegraphics[width = \textwidth]{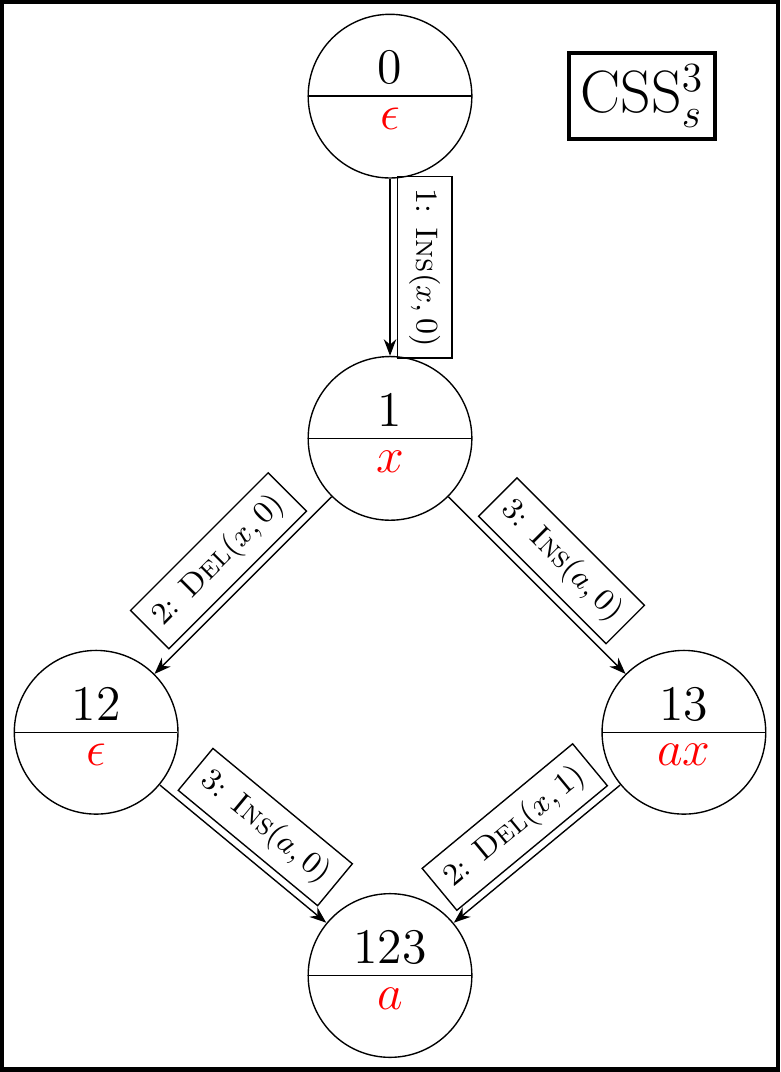}
    \end{subfigure}
    \hfil
    \begin{subfigure}[b]{0.37\textwidth}
      \includegraphics[width = \textwidth]{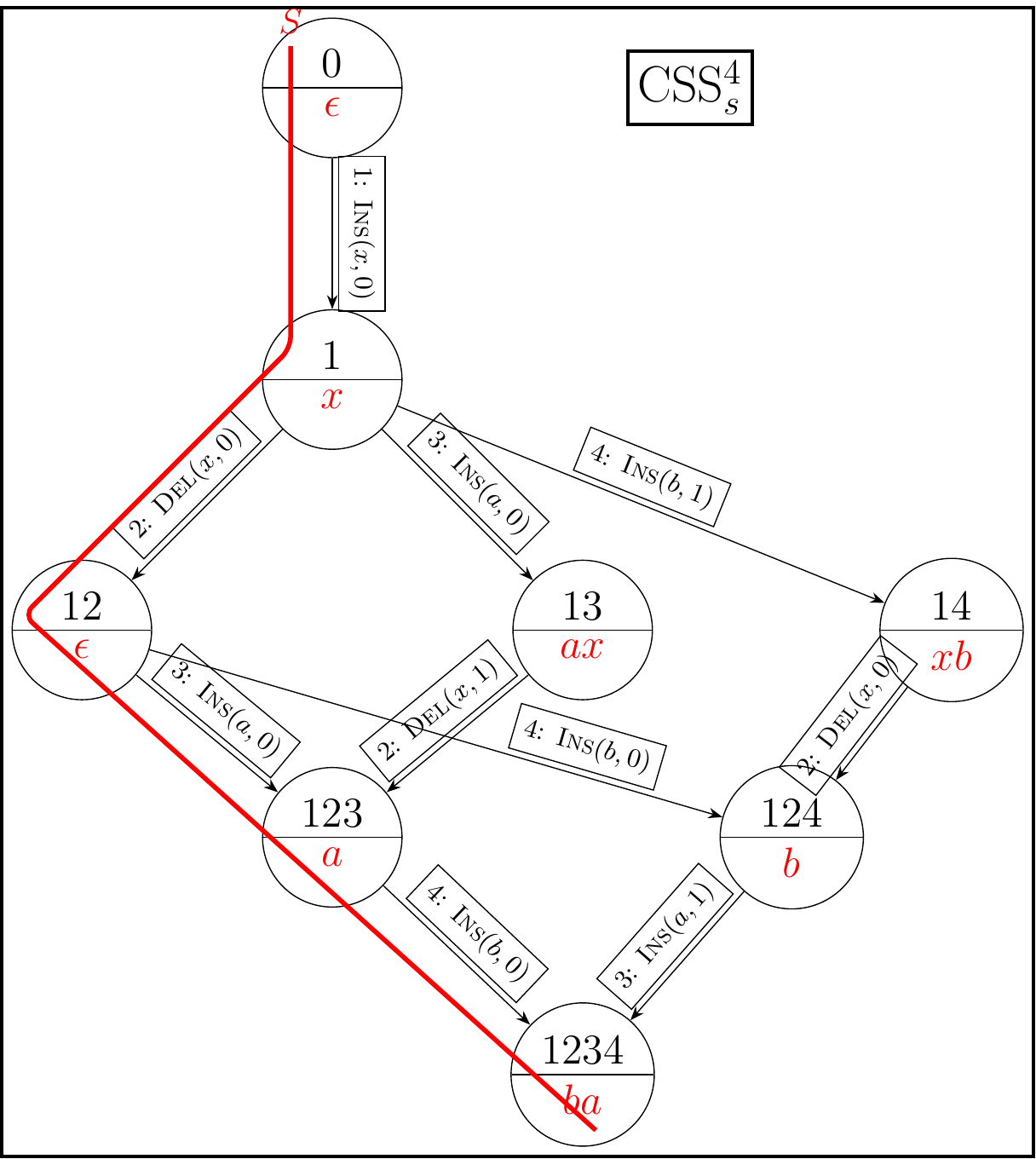}
    \end{subfigure} %
    \caption{Server $s$.}
    \label{fig:cjupiter-illustration-server}
  \end{subfigure} %
  \\[40pt]
  \begin{subfigure}[b]{1.00\textwidth}
    \begin{subfigure}[b]{0.06\textwidth}
      \includegraphics[width = \textwidth]{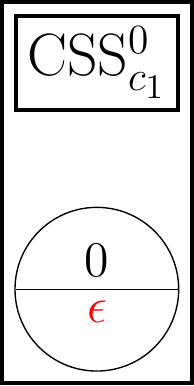}
    \end{subfigure}
    \hfil
    \begin{subfigure}[b]{0.07\textwidth}
      \includegraphics[width = \textwidth]{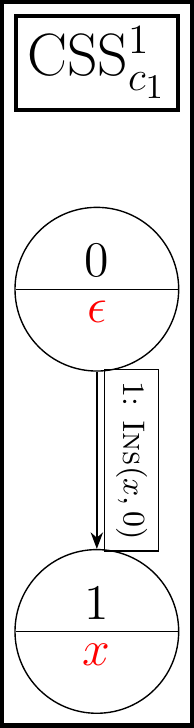}
    \end{subfigure}
    \hfil
    \begin{subfigure}[b]{0.17\textwidth}
      \includegraphics[width = \textwidth]{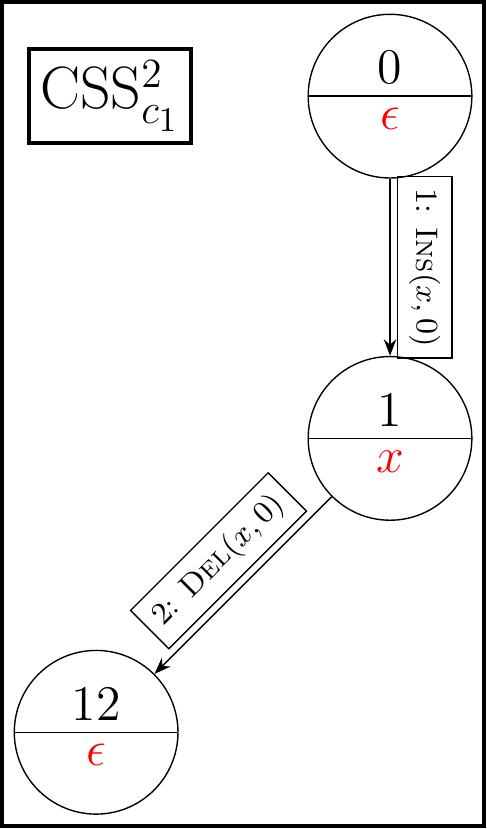}
    \end{subfigure}
    \hfil
    \begin{subfigure}[b]{0.28\textwidth}
      \includegraphics[width = \textwidth]{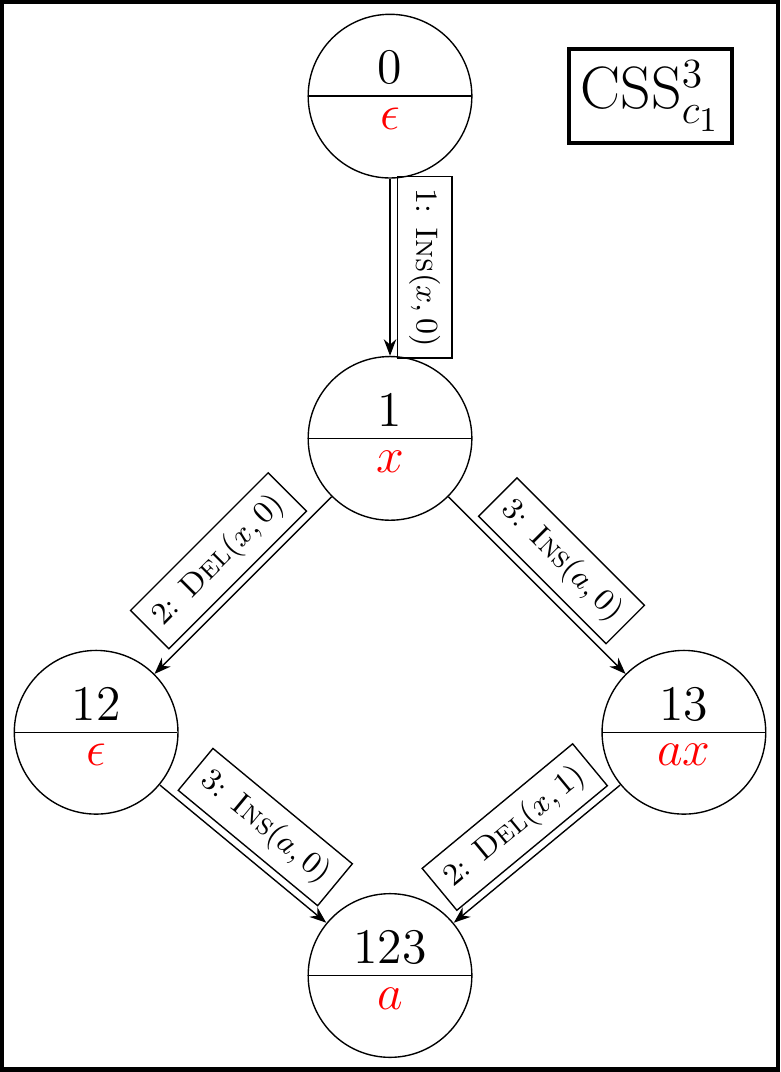}
    \end{subfigure}
    \hfil
    \begin{subfigure}[b]{0.37\textwidth}
      \includegraphics[width = \textwidth]{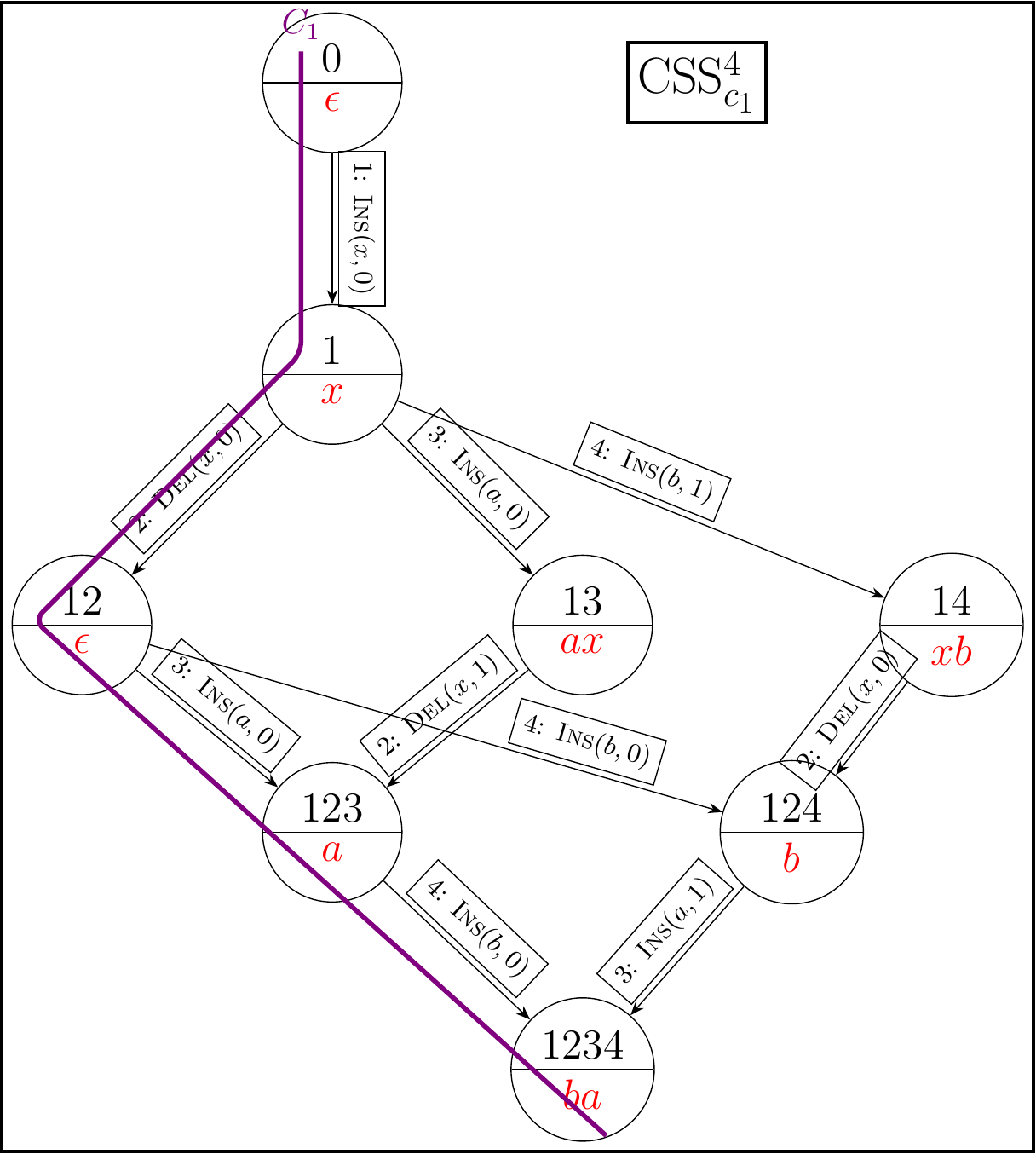}
    \end{subfigure} %
    \caption{Client $c_1$.}
  \end{subfigure} %
  \caption{Illustration of \cjupiter{} under the schedule of Figure~\ref{fig:jupiter-schedule-podc16}.
  The replica behaviors are indicated by the paths in the $n$-ary ordered state spaces. (To be continued)}
\end{figure}
\begin{figure}[t]\ContinuedFloat
  \centering
  \begin{subfigure}[b]{1.00\textwidth}
    \begin{subfigure}[b]{0.06\textwidth}
      \includegraphics[width = \textwidth]{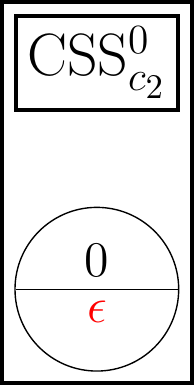}
    \end{subfigure}
    \hfil
    \begin{subfigure}[b]{0.07\textwidth}
      \includegraphics[width = \textwidth]{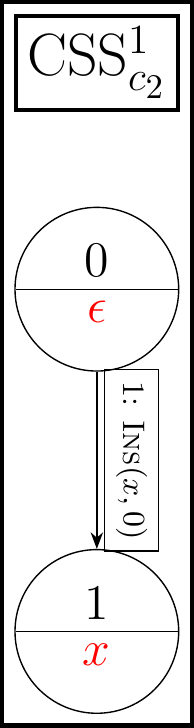}
    \end{subfigure}
    \hfil
    \begin{subfigure}[b]{0.17\textwidth}
      \includegraphics[width = \textwidth]{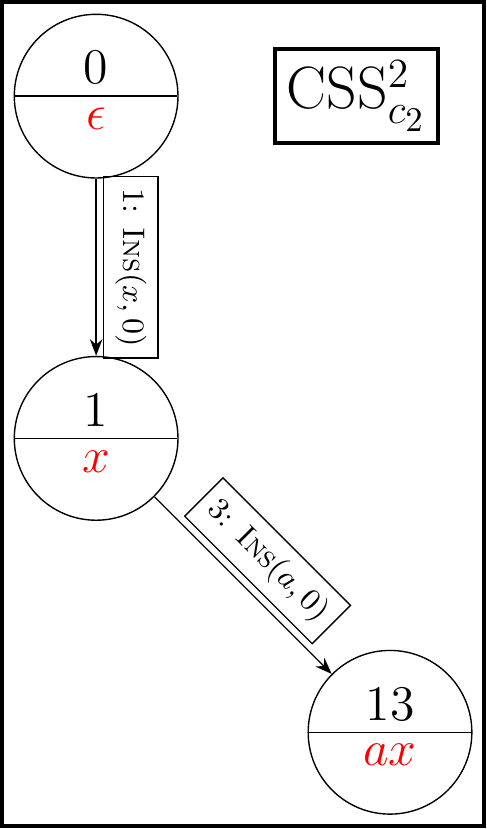}
    \end{subfigure}
    \hfil
    \begin{subfigure}[b]{0.28\textwidth}
      \includegraphics[width = \textwidth]{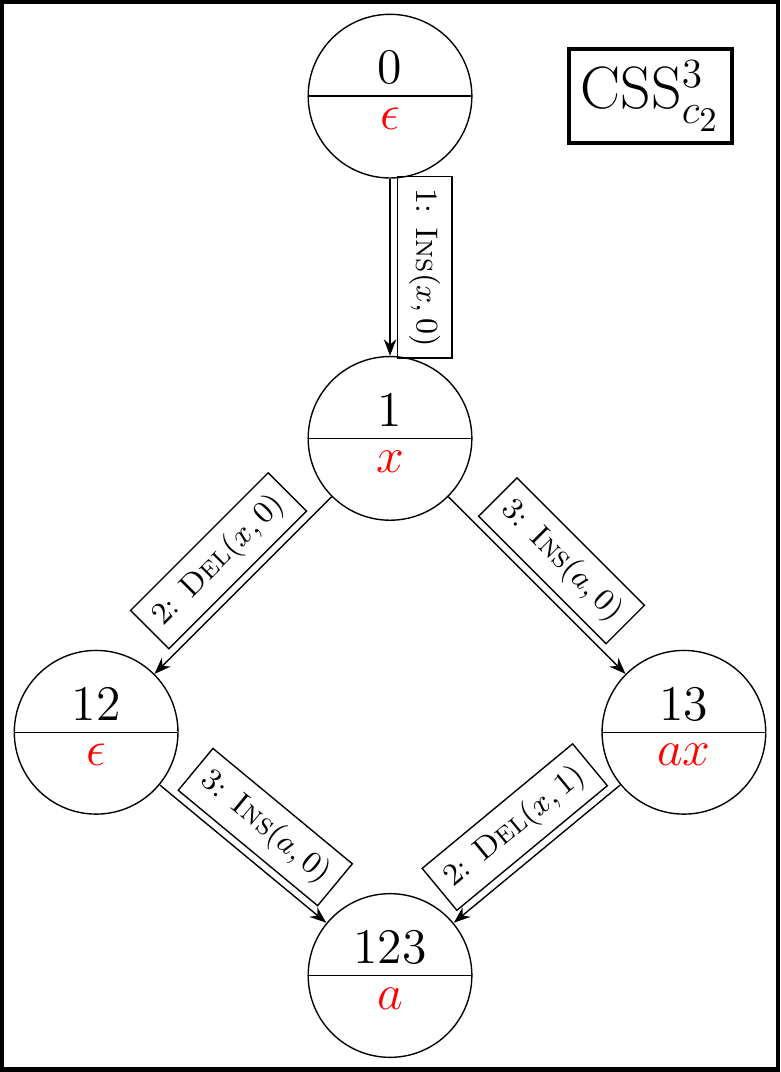}
    \end{subfigure}
    \hfil
    \begin{subfigure}[b]{0.37\textwidth}
      \includegraphics[width = \textwidth]{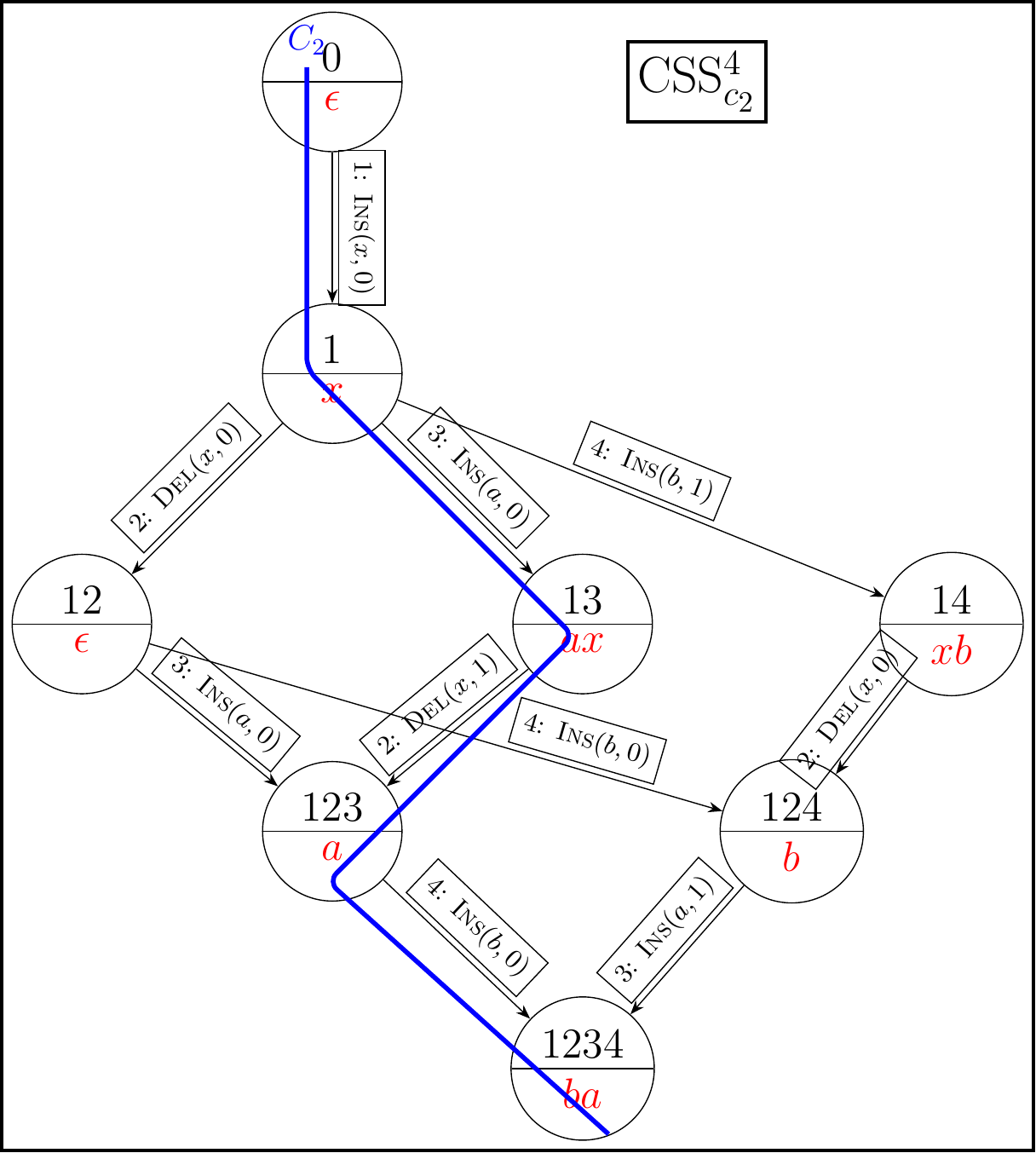}
    \end{subfigure} %
    \caption{Client $c_2$.}
  \end{subfigure} %
  \\[40pt]
  \begin{subfigure}[b]{1.00\textwidth}
    \begin{subfigure}[b]{0.06\textwidth}
      \includegraphics[width = \textwidth]{figs/jupiter-css-podc16-c1423-css0.pdf}
    \end{subfigure}
    \hfil
    \begin{subfigure}[b]{0.07\textwidth}
      \includegraphics[width = \textwidth]{figs/jupiter-css-podc16-c1423-css1.pdf}
    \end{subfigure}
    \hfil
    \begin{subfigure}[b]{0.17\textwidth}
      \includegraphics[width = \textwidth]{figs/jupiter-css-podc16-c1423-css2.pdf}
    \end{subfigure}
    \hfil
    \begin{subfigure}[b]{0.31\textwidth}
      \includegraphics[width = \textwidth]{figs/jupiter-css-podc16-c1423-css3.pdf}
    \end{subfigure}
    \hfil
    \begin{subfigure}[b]{0.35\textwidth}
      \includegraphics[width = \textwidth]{figs/jupiter-css-podc16-c1423-css4.pdf}
    \end{subfigure}
    \caption{Client $c_3$.}
  \end{subfigure}
  \caption{(Continued.) Illustration of \cjupiter{} under the schedule of Figure~\ref{fig:jupiter-schedule-podc16}.
  The replica behaviors are indicated by the paths in the $n$-ary ordered state spaces.}
  \label{fig:appendix-cjupiter-illustration}
\end{figure}

\clearpage


\begin{figure}[hb]
  \centering
  \includegraphics[width = 0.60\textwidth]{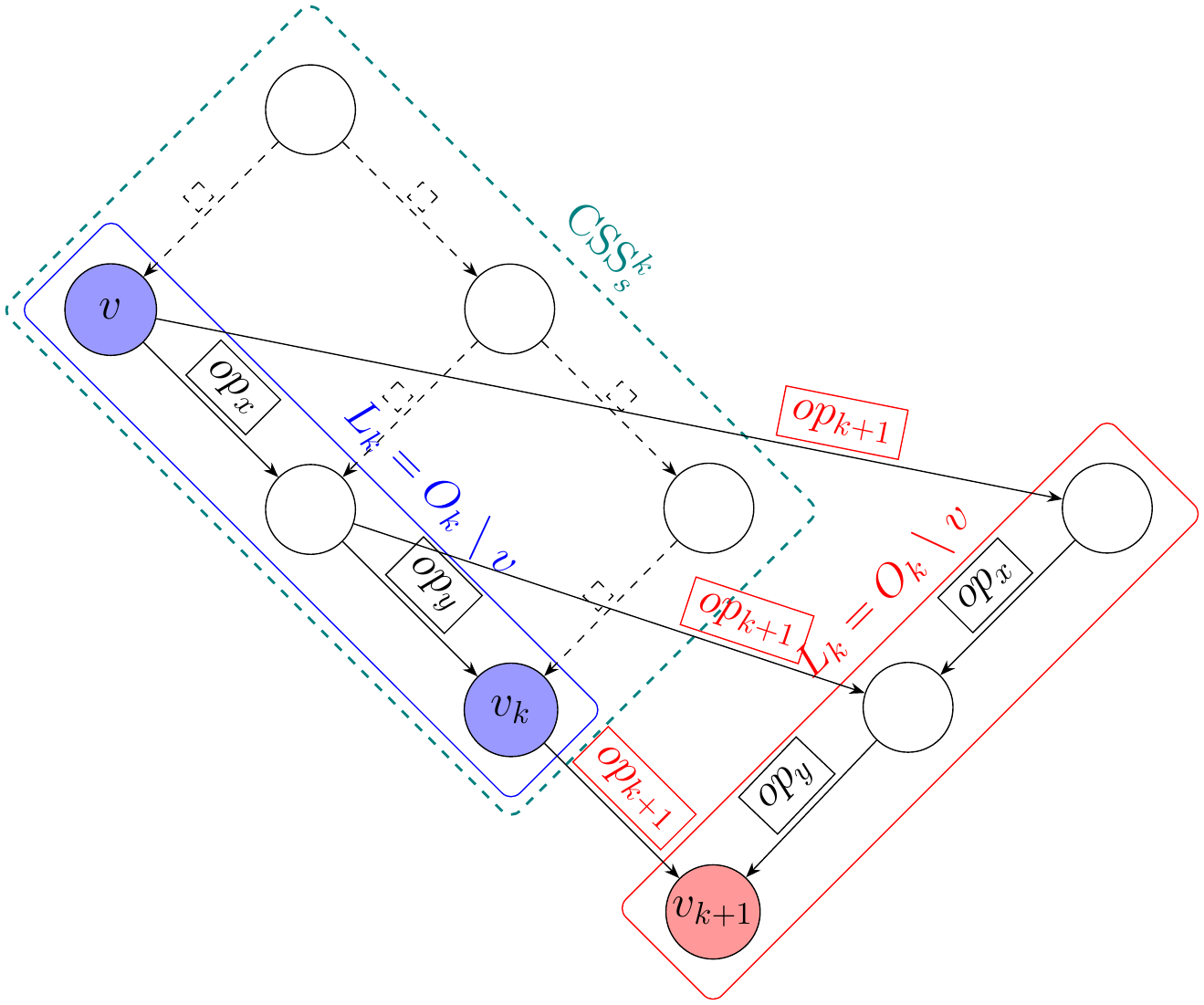}
  \caption{Illustration of \textsc{Case} 2 ($v \neq v_{k}$) of the proof 
    for Lemma~\ref{lemma:cjupiter-first-rule}.}
  \label{fig:cjupiter-first-rule-case2}
\end{figure}

\subsection{Proof for Lemma~\ref{lemma:cjupiter-first-rule} (\cjupiter{}'s ``First'' Rule)}


\begin{proof}
  By mathematical induction on the operation sequence $O$ the server processes.
  
  \emph{Base Case:} $O = \seq{}$.
  \csss{} contains only the initial vertex $v_0 = (\emptyset, \emptyset)$
  and the first edge from $v_0$ is empty.
  
  \emph{Inductive Hypothesis:}
  Suppose that the lemma holds for 
  \[
    O_{k} = \seq{op_1, op_2, \ldots, op_k}.
  \]
  
  \emph{Inductive Step:}
  Consider $O_{k+1} = \seq{op_1, op_2, \ldots, op_{k}, op_{k+1}}$.
  Suppose that the matching vertex of operation $op_{k+1}$ is $v$ (i.e., $v.oids = op_{k+1}.ctx$).
  We distinguish between $v$ being the final vertex of \csssk{k}, denoted $v_{k}$, or not.
  
  \textsc{Case 1:} \emph{$v = v_{k}$.}
  According to the procedure \textsc{xForm} of \cjupiter{} (Algorithm~\ref{alg:css-n-ary-state-space}),
  the state space \csssk{k+1} is obtained by extending \csssk{k}
  with a new edge from $v_k$ labeled with $op_{k+1}$.
  Thus, each path consisting of first edges in \csssk{k} is extended by the edge labeled with $op_{k+1}$,
  meeting the second condition of the lemma in \csssk{k+1}.
  In addition, the first edge from the final vertex of \csssk{k+1} is empty,
  meeting the first condition.
  
  \textsc{Case 2:} \emph{$v \neq v_{k}$.}
  According to the procedure \textsc{xForm} of \cjupiter{} (Algorithm~\ref{alg:css-n-ary-state-space}),
  the server transforms $op_{k+1}$ with the operation sequence, denoted $L_{k}$,
  along the first edges from $v$ to the final vertex $v_{k}$ of \csssk{k},
  obtaining the state space \csssk{k+1} with final vertex $v_{k+1}$.
  By inductive hypothesis, $L_{k}$ consists of the operations
  in $O_{k} \setminus v$ in the total order `$\prec_{s}$'.
  To prove that the lemma holds for \csssk{k+1}, we need to check that
  (Figure~\ref{fig:cjupiter-first-rule-case2}):
  \begin{enumerate}
    \item \emph{It holds for old vertices in \csssk{k}.}
    Each path consisting of first edges from vertices in \csssk{k}
    is extended by the edge labeled with $op_{k+1}$,
    meeting the second condition of the lemma in \csssk{k+1}.
    \item \emph{It holds for new vertices in $\cssskinmath{k+1} \setminus \cssskinmath{k}$.}
    This is because these new vertices form a path
    along which the corresponding operation sequence is exactly $L_{k}$.
  \end{enumerate}
\end{proof}

\subsection{Proof for Lemma~\ref{lemma:cjupiter-ot-server} (\cjupiter{}'s OT Sequence)}

\begin{proof}
  We show that if $L$ is not empty, then
  \begin{enumerate}
    \item \emph{All operations in $L$ are totally ordered by `$\prec_{s}$' before $op$.} 
      This holds because operation $op$ is the last one in the total order `$\prec_s$'.
    \item \emph{All operations in $L$ are concurrent by `$\,\parallel$' with $op$.}
      By contradiction. Suppose that some $op'$ in $L$ is not concurrent with $op$.
      Then it must be the case that $\crel{}{op'}{op}$ and thus $op'$ is not in $L$.
    \item \emph{$L$ consists of all the operations satisfying 2) and 3) and
      all operations in $L$ are totally ordered by `$\prec_{s}$'.}
      This is due to Lemma~\ref{lemma:cjupiter-first-rule}.
  \end{enumerate}
\end{proof}

\subsection{Proof for Proposition~\ref{prop:css-server-client} ($n + 1 \Rightarrow 1$)}

\begin{proof}
  By mathematical induction on the number of operations in the schedule.
  Because all operations are serialized at the server,
  we proceed by mathematical induction on the operation sequence 
  \[
    O = \seq{op_1, op_2, \ldots, op_m} \; (op_i \in \Op{})
  \]
  the server processes in total order `$\prec_{s}$'.
  
  \emph{Base Case.} $O = \seq{op_1}$.
  There is only one operation in the schedule.
  When all replicas have eventually processed this operation,
  they obviously have the same $n$-ary ordered state space.
  Formally,
  \begin{equation*}   \label{eq:css-server-client-bc} 
    \cssskinmath{1} = \cssckinmath{i}{1}, \quad \forall 1 \leq i \leq n.
  \end{equation*}
  
  \emph{Inductive Hypothesis.} $O= \seq{op_1, op_2, \ldots, op_k}$.
  Suppose that when all replicas have eventually processed all the $k$ operations,
  they have the same $n$-ary ordered state space.
  Formally,
  \begin{equation*}   \label{eq:css-server-client-ih} 
    \cssskinmath{k} = \cssckinmath{i}{k}, \quad \forall 1 \leq i \leq n.
  \end{equation*}
  
  \emph{Inductive Step.} $O= \seq{op_1, op_2, \ldots, op_{k+1}}$.
  Suppose that the $(k+1)$-\emph{st} operation $op_{k+1}$ 
  processed at the server is generated by client $c_j$.
  We shall prove that for any client $c_i$,
  when it has eventually processed all these $(k+1)$ operations,
  it has the same $n$-ary ordered state space as the server.
  Formally,
  \begin{equation*}   \label{eq:css-server-client-is} 
    \cssskinmath{k+1} = \cssckinmath{i}{k+1}, \quad \forall 1 \leq i \leq n.
  \end{equation*}
  In the following, we distinguish client $c_j$ that generates $op_{k+1}$ 
  (more specifically, $op_{k+1}.o$ of type $\opset{}$) from other clients.
  
  \textsc{Case 1:} $i \neq j$.
  The $n$-ary ordered state space \csssk{k+1} at the server is obtained by
  applying the $(k+1)$-\emph{st} operation $op_{k+1}$ to \csssk{k}, denoted by
  \begin{equation*}
    \cssskinmath{k+1} = op_{k+1} \otimes \cssskinmath{k}.
  \end{equation*}
  Since the communication is FIFO and in \cjupiter{} the original
  operation (i.e., $op_{k+1}$ here) rather than the transformed one 
  is propagated to clients by the server,
  the $n$-ary ordered state space \cssck{i}{k+1} at client $c_i$ is obtained by
  applying the operation $op_{k+1}$ to \cssck{i}{k}, denoted by
  \begin{equation*}
    \cssckinmath{i}{k+1} = op_{k+1} \otimes \cssckinmath{i}{k}.
  \end{equation*}
  By the inductive hypothesis,
  \begin{equation*}
    \cssskinmath{k} = \cssckinmath{i}{k}, \quad i \neq j.
  \end{equation*}
  Therefore, we have
  \begin{equation*}
    \cssskinmath{k+1} = \cssckinmath{i}{k+1}.
  \end{equation*}
  
  \textsc{Case 2:} $i = j$. 
  Now we consider client $c_j$ that generates the operation $op_{k+1}$.

  Let $\sigma_{k+1}^{c_j} \triangleq \seq{op_{1}^{c_j}, op_{2}^{c_j}, \ldots, op_{k+1}^{c_j}}$,~\footnote{We
  abuse the symbol `$\sigma$' for representing states to denote operation sequences.
  This is reasonable because replica states are defined by the operations a replica has processed
  (Section~\ref{ss:model}).}
  a permutation of $\sigma_{k+1}^{s} \triangleq O$ (i.e., $\seq{op_1, op_2, \ldots, op_{k+1}}$),
  be the operation sequence executed at client $c_j$.
  The operation $op_{k+1}$ may not be the last one executed at client $c_{j}$.
  Instead, suppose $op_{k+1}$ is the $l$-\emph{th} ($1 \leq l \leq k+1$) operation executed at client $c_{j}$,
  namely $op_{l}^{c_j} \equiv op_{k+1}$.
  
  The operation $op_{l}^{c_j}$ splits the sequence $\sigma_{k+1}^{c_j}$ into three parts:
  the subsequence $\sigma_{1, l-1}^{c_j}$ consisting of the first $(l-1)$ operations,
  the subsequence $\sigma_{l,l}^{c_j}$ containing the operation $op_{l}^{c_j} \equiv op_{k+1}$ only,
  and the subsequence $\sigma_{l+1, k+1}^{c_j}$ consisting of the last $(k-l+1)$ operations.
  We formally denote this by
  \[
	  \sigma_{k+1}^{c_j} = \sigma_{1,l-1}^{c_j} \circ op_{k+1} \circ \sigma_{l+1,k+1}^{c_j}.
  \]
  We remark that all operations in $\sigma_{l+1,k+1}^{c_j}$ are concurrent by `$\parallel$' with $op_{k+1}$,
  because they are generated by other clients than $c_j$ before $op_{k+1}$ reaches these clients
  and $op_{k+1}$ is generated before they reach $op_{k+1}$'s local replica (i.e., $c_j$).
  Furthermore, due to the FIFO communication,
  the operations in $\sigma_{l+1,k+1}^{c_j}$ are totally ordered by `$\prec_{s}$'.
  
  Let $\sigma_{k}^{c_j} \triangleq \seq{op_{1}^{c_j}, op_{2}^{c_j}, \ldots, op_{l-1}^{c_j}, op_{l+1}^{c_j},
  \ldots, op_{k+1}^{c_j}}$ be the operation sequence obtained by deleting $op_{l}^{c_j}$
  (i.e., $op_{k+1}$) from $\sigma_{k+1}^{c_j}$, namely
  \[
	  \sigma_{k}^{c_j} = \sigma_{1,l-1}^{c_j} \circ \sigma_{l+1,k+1}^{c_j}.
  \]
  Thus, $\sigma_{k}^{c_j}$ is a permutation of $\sigma_{k}^{s} \triangleq \seq{op_1, op_2, \ldots, op_{k}}$.
  
  In the following, we prove that the $n$-ary ordered state space \cssck{j}{k+1} at client $c_j$
  constructed by executing $\sigma_{k+1}^{c_j}$ in sequence, namely
  \begin{equation*}
    \cssckinmath{j}{k+1} = \sigma_{k+1}^{c_j} \otimes \cssckinmath{j}{0},
  \end{equation*}
  is the same with the $n$-ary ordered state space \csssk{k+1} at the server
  constructed by applying the $(k+1)$-\emph{st} operation $op_{k+1}$ to \csssk{k}, namely
  \begin{equation*}
    \cssskinmath{k+1} = op_{k+1} \otimes \cssskinmath{k}.
  \end{equation*}
  By the inductive hypothesis, \csssk{k} would be the same with
  the $n$-ary ordered state space \cssck{j}{k} constructed at client $c_{j}$ 
  if it had processed $\sigma_{k}^{c_j}$ in sequence.
  Formally,
  \begin{equation*}
    \cssskinmath{k} = \cssckinmath{j}{k} \;(\triangleq \sigma_{k}^{c_j} \otimes \cssckinmath{j}{0}).
  \end{equation*}
  Therefore, it suffices to prove that the $n$-ary ordered state space \cssck{j}{k+1} at client $c_j$
  constructed by executing
  \begin{equation}    \label{eq:s1} 
    \sigma_{k+1}^{c_j} = \sigma_{1,l-1}^{c_j} \circ op_{k+1} \circ \sigma_{l+1,k+1}^{c_j}
  \end{equation}
  in sequence would be the same with the $n$-ary ordered state space constructed at client $c_j$
  if it had processed
  \begin{equation}    \label{eq:s2} 
    \sigma_{k}^{c_j} \circ op_{k+1} = \sigma_{1,l-1}^{c_j} \circ \sigma_{l+1,k+1}^{c_j} \circ op_{k+1}
  \end{equation}
  in sequence.
  
  We first consider the $n$-ary ordered state space obtained by applying $op_{k+1}$ to \cssck{j}{k}
  (which is obtained after executing $\sigma_{k}^{c_j}$) at client $c_j$,
  corresponding to~(\ref{eq:s2}).
  The matching vertex of $op_{k+1}$ is $\sigma_{1,l-1}^{c_j}$.
  According to Lemma~\ref{lemma:cjupiter-ot-server}
  and the inductive hypothesis that \csssk{k} = \cssck{j}{k},
  the operation sequence $L$ with which $op_{k+1}$ transforms
  consists of exactly the (possibly transformed) operations in $\sigma_{l+1, k+1}^{c_j}$:
  \begin{align*}
    L: \;& \copinmath{op_{l+1}^{c_j}}{\sigma_{1,l-1}^{c_j}}, \;
    \copinmath{op_{l+2}^{c_j}}{\sigma_{1,l-1}^{c_j} \circ op_{l+1}^{c_j}}, \;
    \ldots,\\
    & \copinmath{op_{l+3}^{c_j}}{\sigma_{1,l-1}^{c_j} \circ op_{l+1}^{c_j} \circ op_{l+2}^{c_j}}, \;
    \copinmath{op_{k+1}^{c_j}}{\sigma_{1,l-1}^{c_j} \circ op_{l+1}^{c_j} \circ \ldots \circ op_{k}^{c_j}}.
  \end{align*}
  
  We now consider the construction of \cssck{j}{k+1} by executing $\sigma_{k+1}^{c_j}$ in three stages,
  corresponding to~(\ref{eq:s1}).
  \begin{enumerate}
    \item At the beginning, it grows as \cssck{j}{k} does 
      when executing the common subsequence $\sigma_{1,l-1}^{c_j}$.
    \item Next, the operation $op_{k+1}$ is generated at client $c_j$.
    According to the local processing of \cjupiter{},
    the $n$-ary ordered state space grows by saving $op_{k+1}$ at the final vertex 
    (corresponding to) $\sigma_{1,l-1}^{c_j}$ along a new edge.
    \item Then, the sequence $\sigma_{l+1,k+1}^{c_j}$ of operations 
      (from the server) are processed at client $c_{j}$.
    Each operation in $\sigma_{l+1,k+1}^{c_j}$, when executed in sequence,
    not only ``simulates'' the growth of \cssck{j}{k},
    but also completes one step of the iterative operational transformations of $op_{k+1}$
    with the sequence $L$ mentioned above when applying $op_{k+1}$ to \cssck{j}{k}.
    (This can be proved by mathematical induction.)
    We take as an example the case of the first operation $op_{l+1}^{c_j}$.
    After transforming with some subsequence of operations (which may be empty)
    in $\sigma_{1,l-1}^{c_j}$,
    operation $op_{l+1}^{c_j}$ is transformed as $\copinmath{op_{l+1}^{c_j}}{\sigma_{1,l-1}^{c_j}}$.
    At that time, \cop{op_{l+1}^{c_j}}{\sigma_{1,l-1}^{c_j}} is then transformed with
    \cop{op_{k+1}}{\sigma_{1,l-1}^{c_j}}, which is also performed when applying $op_{k+1}$ to \cssck{j}{k}:
    \begin{align*}
      &OT(\copinmath{op_{l+1}^{c_j}}{\sigma_{1, l-1}^{c_j}},
      \copinmath{op_{k+1}}{\sigma_{1, l-1}^{c_j}}) \\
      = \big(&\copinmath{op_{l+1}^{c_j}}{\sigma_{1,l-1}^{c_j} \circ op_{k+1}},
      \copinmath{op_{k+1}}{\sigma_{1,l-1}^{c_j} \circ op_{l+1}^{c_j}}\big).
    \end{align*}
    As it goes on, after executing $\sigma_{k+1}^{c_j}$ in sequence,
    we obtain an $n$-ary ordered state space 
    same with that obtained by applying $op_{k+1}$ to \cssck{j}{k}.
  \end{enumerate}
\end{proof}

\clearpage
\setcounter{figure}{0}
\section{The \jupiter{} Protocol}  \label{appendix:section-jupiter}

We review the \jupiter{} protocol in~\cite{Xu:CSCW14},
a \emph{multi-client} description of \jupiter{} first proposed in~\cite{Nichols:UIST95}.

\subsection{Data Structure: 2D State Space}  \label{appendix:ss-2d-state-space}

For a client/server system with $n$ clients,
\jupiter{} maintains $2n$ 2D state spaces,
each of which consists of a local dimension and a global dimension.
We first define operations and vertices as follows.

\begin{definition}[Operation]	\label{def:jupiter-op}
  Each operation $op$ of type $\Op$ (Algorithm~\ref{alg:jupiter-op}) is a tuple $op = (o, oid, ctx)$, where 
  \begin{itemize}
    \item $o:$ the signature of type $\opset$ described in Section~\ref{ss:list-spec};
    \item $oid:$ a globally unique identifier
      which is a pair $(cid, seq)$ 
      consisting of the client id and a sequence number; \emph{and}
    \item $ctx:$ an \emph{operation context} which is a set of operation identifiers,
      denoting the operations that are causally before $op$.
  \end{itemize}
\end{definition}

The OT functions of two operations $op, op' \in \Op$,
\begin{gather*}
  OT: \Op{} \times \Op{} \to \Op{} \times \Op{} \\
  (\opot{op}{op'}, \opot{op'}{op}) = OT(op, op'),
\end{gather*}
are defined based on those of operations $op.o, op'.o \in \opset$,
denoted $(o, o') = OT(op.o, op'.o)$,
such that
\begin{align*}
  \opot{op}{op'} &= (o, op.oid, op.ctx \cup \set{op'.oid}), \\
  \opot{op'}{op} &= (o', op'.oid, op'.ctx \cup \set{op.oid}).
\end{align*}

A 2D state space is a finite set of vertices.

\begin{definition}[Vertex]	\label{def:jupiter-vertex}
  A vertex $v$ of type \Vertex{} (Algorithm~\ref{alg:jupiter-vertex}) 
  is a pair $v = (oids, edges)$, where
  \begin{itemize}
    \item $oids \in 2^{\mathbb{N}_{0} \times \mathbb{N}_{0}}$
      is the set of operations (represented by their identifies) that have been executed.
    \item $edges$ is an array of \emph{two} (indexed by \Local{} and \Remote{}) 
      edges of type \Edge{} (Algorithm~\ref{alg:jupiter-edge})
      from $v$ to two other vertices, labeled with operations.
      That is, each edge is a pair $(op: \Op, v: \Vertex)$.
  \end{itemize}
\end{definition}

For vertex $u$, we say that $u.edges[\Local].op$ is an operation 
from $u$ along the \emph{local} dimension/edge
and $u.edges[\Remote].op$ along the \emph{remote} dimension/edge.
This is similar for the child vertices $u.edges[\Local].v$
and $u.edges[\Remote].v$ of $u$.

As with in an $n$-ary ordered state space, 
for each vertex $v$ and each edge $e$ from $v$ in a 2D state space, 
it is required that
\begin{itemize}
  \item the $ctx$ of the operation $e.op$ associated with $e$ matches the $oids$ of $v$: 
    $e.op.ctx = v.oids$.
  \item the $oids$ of the vertex $e.v$ along $e$ 
    consists of the $oids$ of $v$ and the $oid$ of $e.op$:
    $e.v.oids = v.oids \cup \set{e.op.oid}$.
\end{itemize}

\clearpage
\begin{algorithm}
  \caption{Operation in \jupiter{}.}
  \label{alg:jupiter-op}
  \begin{algorithmic}[1]
    \CLASS{Op} 
      \State \VAR{o}\, : $\opset$
      \State \VAR{oid} : $\CID{} \times \SEQ{}$
      \State \VAR{ctx} : $2^{\CID{} \,\times\, \SEQ{}} = \emptyset$

      \Statex
      \Procedure{OT}{$op : \Op, op' : \Op$} : ($\Op, \Op$)
	\State $(o, o') \gets \Call{OT}{op.o, op.o'}$ \Comment{call OT on $\opset$}

	\hStatex
	\State \algparbox{$\Op\; \opot{op}{op'} = \text{ new } \Op(o, op.oid, op.ctx \cup \set{op'.oid})$}
	\State \algparbox{$\Op\; \opot{op'}{op} \!=\! \text{ new } \Op(o', op'.oid, op'.ctx \cup \set{op.oid})$}

	\hStatex
	\State \Return $(\opot{op}{op'}, \opot{op'}{op})$
      \EndProcedure
    \EndClass{Op}
  \end{algorithmic}
\end{algorithm}

\begin{algorithm}
  \caption{Vertex in the 2D state space.}
  \label{alg:jupiter-vertex}
  \begin{algorithmic}[1]
    \CLASS{Vertex} 
      \State \VAR{oids} : $2^{\CID{} \,\times\, \SEQ{}} = \emptyset$
      \State \VAR{edges} : $\Edge[2] = \set{[\Local] = [\Remote] = \Null}$
    \EndClass{Vertex}
  \end{algorithmic}
\end{algorithm}


\begin{algorithm}
  \caption{Edge in the 2D state space.}
  \label{alg:jupiter-edge}
  \begin{algorithmic}[1]
    \CLASS{Edge}
      \State \VAR{op} : $\Op = \Null{}$
      \State \VAR{v} : $\Vertex = \Null{}$
    \EndClass{Edge}
  \end{algorithmic}
\end{algorithm}
\begin{algorithm}
  \caption{2D state space.}
  \label{alg:jupiter-2d-state-space}
  \begin{algorithmic}[1]
    \CLASS{StateSpace2D}
      \State \VAR{cur}\, : \Vertex = new \Vertex()

      \Statex
      \Procedure{xForm}{$op : \Op, d : \LR$} : \Op
	\State \Vertex{} $u \gets$ \Call{Locate}{$op$}
	\State \Vertex{} $v \gets$ \Call{Add}{$op, 1-d, u$}

	\hStatex
	\While {$u \neq cur$} \Comment{See Figure~\ref{fig:xform-ot}}
	  \State \Vertex{} $u' \gets u.edges[d].v$
	  \State $\Op{}\; op' \gets u.edges[d].op$

	  \hStatex
	  \State $(\opot{op}{op'}, \opot{op'}{op}) \gets$ \Call{OT}{$op, op'$}

	  \hStatex
	  \State \Vertex{} $v'$ = new \Vertex($v.oids \cup \set{op'.oid}, \emptyset$)
	  \State \Edge{} $e_{vv'} \gets$ new \Edge($\opot{op'}{op}, v'$)
	  \State $v.edges[d] \gets e_{vv'}$
	  \State \Edge{} $e_{u'v'} \gets$ new \Edge($\opot{op}{op'}, v'$)
	  \State $u'.edges[1-d] \gets e_{u'v'}$

	  \hStatex
	  \State $u \gets u'$
	  \State $v \gets v'$
	  \State $op \gets \opot{op}{op'}$
	\EndWhile

	\hStatex
	\State $cur \gets v$
	\State \Return $op$
      \EndProcedure

      \Statex
      \Procedure{Locate}{$op : \Op$} : \Vertex{}
	\State \Return \Vertex{} $v$ with $v.oids = op.ctx$
      \EndProcedure

      \Statex
      \Procedure{Add}{$op : \Op, d : \LR, u : \Vertex$} : \Vertex{}
	\State \Vertex{} $v \gets \text{new } \Vertex(u.oids \cup \set{op.oid}, \emptyset)$

	\hStatex
	\State \Edge{} $e \gets$ new \Edge($op, v$)
	\State $u.edges[d] \gets e$

	\hStatex
	\State \Return $v$
      \EndProcedure
    \EndClass{StateSpace2D}
  \end{algorithmic}
\end{algorithm}

\begin{definition}[2D State Space]  \label{def:jupiter-2d-state-space}
  A set of vertices $S$ is a 2D state space if and only if
  \begin{enumerate}
    \item Vertices are uniquely identified by their $oids$.
    \item For each vertex $u$ with $|u.edges| = 2$,
      let $u'$ be its child vertex along the local dimension/edge $e_{uu'} = (op', u')$
      and $v$ the other child vertex along the global dimension/edge $e_{uv} = (op, v)$.
      There exist (Figure~\ref{fig:xform-ot})
      \begin{itemize}
	\item a vertex $v'$ with $v'.oids = u.oids \cup \set{op'.oid, op.oid}$;
	\item an edge $e_{u'v'} = (\opot{op}{op'}, v')$ from $u'$ to $v'$;
	\item an edge $e_{vv'} = (\opot{op'}{op}, v')$ from $v$ to $v'$.
      \end{itemize}
  \end{enumerate}
\end{definition}

The second condition above models OTs in \jupiter{}.

\clearpage
\subsection{The \jupiter{} Protocol}  \label{appendix:ss-jupiter}

Each client $c_i$ maintains a 2D state space, denoted \cscwc{i},
with the local dimension for operations generated by the client
and the global dimension for operations generated by other clients.
The server maintains $n$ 2D state spaces, one for each client.
The state space for client $c_i$, denoted \cscws{i},
consists of the local dimension for operations from client $c_i$
and the global dimension for operations from other clients.

\jupiter{} is similar to \cjupiter{} with two major differences:
\begin{enumerate}
  \item In $\textsc{xForm}(op: \Op, d: \LR = \set{\Local, \Remote})$ of \jupiter{}, 
    the operation sequence with which $op$ transforms is determined by an extra parameter $d$; \emph{and}
  \item In \jupiter{}, the server propagates the transformed operation 
    (instead of the original one it receives from a client) to other clients.
\end{enumerate}

As with \cjupiter{}, we also describe \jupiter{} in three parts.
In the following, we omit the details that are in common with and have been explained in \cjupiter{}.
\subsubsection{Local Processing (\textsc{Do} of Algorithm~\ref{alg:jupiter-client})}   \label{sss:jupiter-local}

When client $c_i$ receives an operation $o \in \opset$ from a user, it
\begin{enumerate}
  \item applies $o$ locally;
  \item generates $op \in \Op$ for $o$ 
    and saves it along the local dimension at the end of its 2D state space \cscwc{i}; \emph{and}
  \item sends $op$ to the server.
\end{enumerate}
\subsubsection{Server Processing (\textsc{Receive} of Algorithm~\ref{alg:jupiter-server})}   \label{sss:jupiter-server}

When the server receives an operation $op \in \Op$ from client $c_i$, it

\begin{enumerate}
  \item transforms $op$ with an operation sequence along the global dimension 
    in the 2D state space \cscws{i} to obtain $op'$
    by calling $\textsc{xForm}(op, \Remote)$ (Section~\ref{sss:2d-state-space-xform});
  \item applies $op'$ locally;
  \item for each $j \neq i$, saves $op'$ at the end of \cscws{j} along the global dimension; \emph{and}
  \item sends $op'$ (instead of $op$) to other clients.
\end{enumerate}
\subsubsection{Remote Processing (\textsc{Receive} of Algorithm~\ref{alg:jupiter-client})}   \label{sss:jupiter-remote}

When client $c_i$ receives an operation $op \in \Op$ from the server, it

\begin{enumerate}
  \item transforms $op$ with an operation sequence along the local dimension 
    in its 2D state space \cscwc{i} to obtain $op'$
    by calling $\textsc{xForm}(op, \Local)$ (Section~\ref{sss:2d-state-space-xform}); \emph{and}
  \item applies $op'$ locally.
\end{enumerate}
\subsubsection{OTs in \jupiter{} (\textsc{xForm} of Algorithm~\ref{alg:jupiter-2d-state-space})}   \label{sss:2d-state-space-xform}

The procedure $\textsc{xForm}(op: \Op, d: \LR = \set{\Local, \Remote})$ of \jupiter{} 
is similar to $\textsc{xForm}(op: \Op)$ of \cjupiter{}
except that in \jupiter{}, 
the operation sequence with which $op$ transforms is determined by an extra parameter $d$.
Specifically, it

\begin{enumerate}
  \item locates the vertex $u$ whose $oids$ matches the operation context $op.ctx$ of $op$; \emph{and}
  \item iteratively transforms $op$ with an operation sequence along the $d$ dimension from $u$
    to the final vertex $cur$ of this 2D state space.
\end{enumerate}

\begin{algorithm}
  \caption{Client in \jupiter{}.}
  \label{alg:jupiter-client}
  \begin{algorithmic}[1]
    \CLASS{Client}
      \State \VAR{cid} : \CID{}
      \State \VAR{seq} : \SEQ{} = 0
      \State \VAR{state} : $\Sigma = \emptyseq$
      \State \VAR{S} : \StateSpace = new \StateSpace()

      \Statex
      \Procedure{do}{$o : \opset$} : \Val{} \Comment{Local Processing}
	\State $(state, val) \gets$ \Call{Apply}{$state, o$}

	\hStatex
	\State $seq \gets seq + 1$
	\State \Op{} $op \gets$ new $\Op(o, (cid, seq), S.cur.oids)$

	\hStatex
	\State \Vertex{} $v \gets S$.\Call{Add}{$op, \Local, S.cur$}
	\State $S.cur \gets v$

	\hStatex
	\State \Call{Send}{$\SID{}, op$}  \Comment{send $op$ to the server}

	\hStatex
	\State \Return $val$
      \EndProcedure

      \Statex
      \Procedure{receive}{$op : \Op$} \Comment{Remote Processing}
	\State $\Op\; op' \gets S.\Call{xForm}{op, \Local}$
	\State $state \gets state \circ op'.o$
      \EndProcedure
    \EndClass{Client}
  \end{algorithmic}
\end{algorithm}

\begin{algorithm}
  \caption{Server in \jupiter{}.}
  \label{alg:jupiter-server}
  \begin{algorithmic}[1]
    \CLASS{Server}
      \State \VAR{SS} : \StateSpace[\CID{}]	\Comment{one per client}
      \State \VAR{state} : $\Sigma = \emptyseq$

      \Statex
      \Procedure{Receive}{$op: \Op$} \Comment{Server Processing}
	\State $\Op\; op' \gets SS[op.oid.cid].\Call{xForm}{op, \Remote}$
	\State $state \gets state \circ op'.o$

	\hStatex
	\ForAll{$c \in \CID{} \setminus \set{op.oid.cid}$}
	  \State $SS[c].\Call{Add}{op, \Remote{}, SS[c].cur}$
	  \State \Call{Send}{$c, op'$}	\Comment{send $op'$ (not $op$) to client $c$}
	\EndFor
      \EndProcedure
    \EndClass{Server}
  \end{algorithmic}
\end{algorithm}

\begin{figure}[t]
  \begin{sideways}
   \includegraphics[width = 1.30\textwidth]{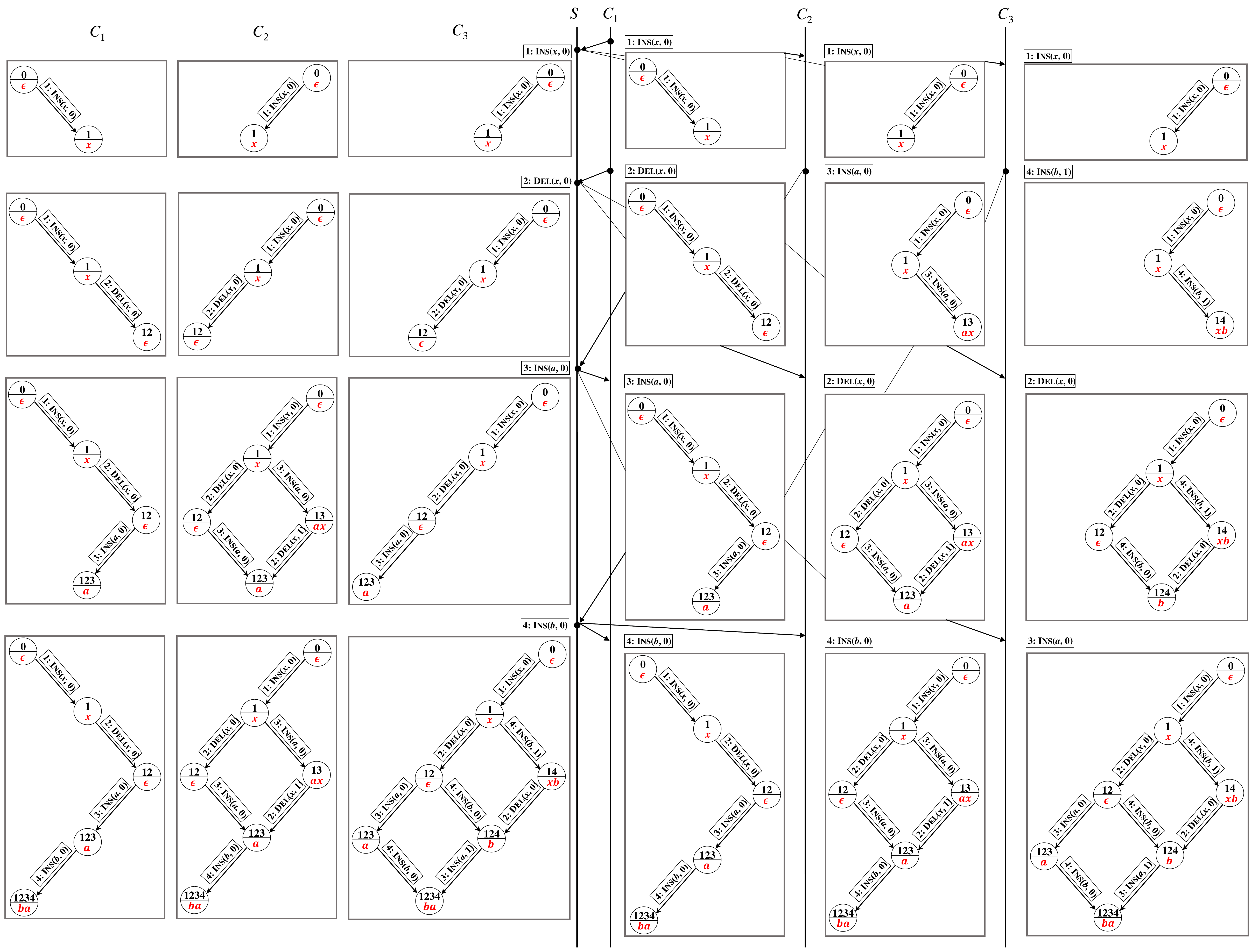}
  \end{sideways}
  \centering
  \caption{(Rotated) illustration of \jupiter{}~\cite{Xu:CSCW14} under the schedule of Figure~\ref{fig:jupiter-schedule-podc16}.}
  \label{fig:cjupiter-illustration}
\end{figure}

\clearpage
\setcounter{figure}{0}
\section{Proofs for Section~\ref{section:jupiter-cjupiter-equiv}: \cjupiter{} is Equivalent to \jupiter{}}   \label{appendix:jupiter-cjupiter-equiv-proof}

\subsection{Proof for Proposition~\ref{prop:server-equiv} ($n \leftrightarrow 1$)}

\begin{proof}
  By mathematical induction on the operation sequence $O = \seq{op_1, op_2, \cdots, op_m}$ 
  the server processes.
  
  \emph{Base Case.} $k=1$.
  According to the \jupiter{} and \cjupiter{} protocols, it is obviously that
  \[
    \cssskinmath{1} = \cscwskinmath{c(op_{1})}{1}.
  \]
  
  \emph{Inductive Hypothesis.}
  Suppose that \eqref{eq:server-equiv} holds for $k$:
  \[
    \cssskinmath{k} = \bigcup_{i=1}^{i=k} \cscwskinmath{c(op_{i})}{i}.
  \]
  
  \emph{Inductive Step.}
  We shall prove that \eqref{eq:server-equiv} holds for $(k+1)$:
  \[
    \cssskinmath{k+1} = \bigcup_{i=1}^{i=k+1} \cscwskinmath{c(op_{i})}{i}.
  \]
  
  By inductive hypothesis, we shall prove that
  \[
    \cssskinmath{k+1} \setminus \cssskinmath{k} =  \cscwskinmath{c(op_{k+1})}{k+1}.
  \]
  In other words, the OTs for $op_{k+1}$ performed by the servers 
  in \jupiter{} and \cjupiter{} are the same.
  This holds due to two reasons.
  First, under the same schedule, the matching vertex of $op_{k+1}$ 
  in \cscwsk{c(op_{k+1})}{k} of \jupiter{}
  is the same with that in \csssk{k} of \cjupiter{},
  determined by its operation context (or the causally-before relation of the schedule).
  Second, according to Lemma~\ref{lemma:cjupiter-ot-server} for \cjupiter{}
  and its counterpart for \jupiter{},
  the operation sequences with which $op_{k+1}$ transforms are the same in both protocols.
\end{proof}

\subsection{Proof for Proposition~\ref{prop:client-equiv} ($1 \leftrightarrow 1$)} \label{appendix:pf-client-equiv}

\begin{proof}
  By mathematical induction on the operation sequence
  $O^{c_i} = \seq{op_{1}^{c_i}, op_{2}^{c_i}, \ldots, op_{m}^{c_i}}$ the client $c_i$ processes.
  
  \emph{Base case.} $k = 1$, namely, $O^{c_i} = \seq{op_{1}^{c_i}}$.
  No matter whether $op_{1}^{c_i}$ (more specifically, $op_{1}^{c_i}.o$) is generated by client $c_i$
  or is an operation propagated to client $c_i$ by the server,
  it obviously holds that
  \begin{equation*}
    \cscwckinmath{i}{1} = \cssckinmath{i}{1}.
  \end{equation*}
  
  \emph{Inductive Hypothesis.} Suppose $O^{c_i} = \seq{op_{1}^{c_i}, op_{2}^{c_i}, \ldots, op_{k}^{c_i}}$ and
  \eqref{eq:client-equiv} holds for $k$:
  \begin{equation*}
    \cscwckinmath{i}{k} \subseteq \cssckinmath{i}{k}.
  \end{equation*}
  
  \emph{Inductive Step.} Client $c_i$ executes the $(k+1)$-\emph{st} operation $op_{k+1}^{c_i}$.
  We shall prove that \eqref{eq:client-equiv} holds for $(k+1)$:
  \begin{equation*}
    \cscwckinmath{i}{k+1} \subseteq \cssckinmath{i}{k+1}.
  \end{equation*}
  We distinguish two cases between $op_{k+1}^{c_i}$ being generated by client $c_i$
  or an operation propagated to client $c_i$ by the server.
  
  \textsc{Case 1:} \emph{The operation $op_{k+1}^{c_i}$ is generated by client $c_i$.}
  The new 2D state space \cscwck{i}{k+1} of \jupiter{} 
  (resp. $n$-ordered state space \cssck{i}{k+1} of \cjupiter{})
  is obtained by saving $op_{k+1}^{c_i}$ at the final vertex of the previous state space \cscwck{i}{k}
  (resp. \cssck{i}{k}).
  Since $\cscwckinmath{i}{k} \subseteq \cssckinmath{i}{k}$ (by the inductive hypothesis),
  we conclude that $\cscwckinmath{i}{k+1} \subseteq \cssckinmath{i}{k+1}$.
  
  \textsc{Case 2:} \emph{The operation $op_{k+1}^{c_i}$ is an operation 
  propagated to client $c_i$ by the server.}
  Due to Lemmas~\ref{lemma:cjupiter-ot-server} for \cjupiter{} and its counterpart for \jupiter{},
  the operation sequences $L$ with which $op_{k+1}^{c_i}$ transforms
  at the server in both protocols are the same.
  Since the communication is FIFO, when client $c_i$ receives $op_{k+1}^{c_i}$,
  all the operations totally ordered by `$\prec_{s}$' before $op_{k+1}^{c_i}$ have already been in \cssck{i}{k}.
  By Proposition~\ref{prop:css-server-client}, 
  the OTs involved in iteratively transforming $op_{k+1}^{c_i}$ with $L$
  at the server in both protocols
  are also performed at client $c_i$ in \cjupiter{}.
  By contrast, in \jupiter{}, the resulting transformed operation, 
  denoted $\opot{op_{k+1}^{c_i}}{L}$,
  is propagated to client $c_i$, where the set of OTs performed is a subset of those
  involved in transforming $op_{k+1}^{c_i}$ with $L$.
  Given the inductive hypothesis $\cscwckinmath{i}{k} \subseteq \cssckinmath{i}{k}$,
  we conclude that $\cscwckinmath{i}{k+1} \subseteq \cssckinmath{i}{k+1}$.
\end{proof}

\subsection{Proof for Theorem~\ref{thm:client-equiv} (Equivalence of Clients)} \label{appendix:proof-client-equiv-theorem}

\begin{proof}
  Note that in the proof for Proposition~\ref{prop:client-equiv},
  no matter whether the operation $op_{k}^{c_i}$ is generated by client $c_i$
  or is an operation propagated to client $c_i$ by the server,
  the final transformed operations executed at $c_i$ in \jupiter{} and \cjupiter{} are the same.
\end{proof}

\clearpage
\setcounter{figure}{0}
\section{Proofs for Section~\ref{section:cjupiter-weak-spec}: \cjupiter{} Satisfies the Weak List Specification} \label{section:appendix-wlspec}

\subsection{Proof for Lemma~\ref{lemma:simple-path} (Simple Path)}    \label{ss:appendix-simple-path}

\begin{proof}
  Due to the specific structure of OTs (Figure~\ref{fig:xform-ot}), 
  all the transitions associated with the same operation are ``parallel'' 
  in $n$-ary ordered state spaces. 
  They cannot be in the same path.
\end{proof}

\subsection{Proof for Lemma~\ref{lemma:irreflexivity} (Irreflexivity)}  \label{ss:proof-irreflexivity}

\begin{proof}
  We prove both directions by contradiction.
  
  \emph{``$\Leftarrow$'' (if):}
  Suppose by contradiction that $\lorel{a}{a}$ for some $a \in \elems{H}$.
  According to Lemma~\ref{lemma:simple-path}, a list state $w$ contains no duplicate elements.
  Therefore, there exist two list states such that for some element $b$,
  $\lorel{a}{b}$ (namely, $a$ precedes $b$) in one state
  and $\lorel{b}{a}$ (namely, $b$ precedes $a$) in the other.
  However, this contradicts the assumption that all list states are pairwise compatible.
  
  \emph{``$\Rightarrow$`` (only if):}
  Suppose by contradiction that two list states $w_1$ and $w_2$ are incompatible.
  That is, they have two common elements $a$ and $b$ such that
  $a$ precedes $b$ in, say, $w_1$ and $b$ precedes $a$ in $w_2$.
  Thus, both $\lorel{a}{b}$ and $\lorel{b}{a}$ hold.
  Since \lo{} is transitive on $w_1$ (and $w_2$), we have $\lorel{a}{a}$,
  contradicting the assumption that \lo{} is irreflexive.
\end{proof}

\clearpage
\subsection{Proof for Lemma~\ref{lemma:lca} (LCA)}    \label{ss:appendix-lca}

\begin{proof}
  By mathematical induction on the operation sequence $O = \seq{op_1, op_2, \cdots, op_m}$ ($op_i \in \Op{}$)
  processed in total order `$\prec_{s}$' at the server.
  
  \emph{Base Case.} Initially, the $n$-ary ordered state space \csssk{0} at the server
  contains only the single initial vertex $v_0 = (\emptyset, \emptyset)$.
  The lemma obviously holds.
  
  \emph{Inductive Hypothesis.} Suppose that the server has processed $k$ operations and that
  every pair of vertices in the $n$-ary ordered state space \csssk{k} has a unique LCA.
  
  \emph{Inductive Step.} The server has processed the $(k+1)$-\emph{st} operation $op_{k+1}$.
  We shall prove that every pair of vertices in the $n$-ary ordered state space \csssk{k+1} 
  has a unique LCA.
  Let
  \begin{equation*}
    \text{CSS}_{\Delta} \triangleq \cssskinmath{k+1} \setminus \cssskinmath{k}
  \end{equation*}
  be the extra part of \csssk{k+1} obtained by transforming $op_{k+1}$
  with some operation sequence, denoted $L$, in \csssk{k} (Figure~\ref{fig:lca}).
  We need to verify that
  \itshape 1\upshape) every pair of vertices in $\text{CSS}_{\Delta}$ has a unique LCA; \emph{and}
  \itshape 2\upshape) every pair of vertices consisting of one vertex in \csssk{k} 
      and the other in $\text{CSS}_{\Delta}$ has a unique LCA.
  
  The former claim obviously holds because all vertices in $\text{CSS}_{\Delta}$ are in a path.
  We prove the latter by contradiction.
  Let $v_1$ be any vertex in \csssk{k} and $v_2$ any vertex in $\text{CSS}_{\Delta}$ 
  (Figure~\ref{fig:lca}).
  Clearly, the initial vertex $v_0 = (\emptyset, \emptyset)$ is a common ancestor of $v_1$ and $v_2$.
  Suppose by contradiction that there are two LCAs,
  denoted $v$ and $v'$, of $v_1$ and $v_2$ in \csssk{k}
  (they cannot be in $\text{CSS}_{\Delta}$).
  
  Note that any path from $v$ or $v'$ 
  to $v_2$ passes through some vertex in the operation sequence $L$ with which $op_{k+1}$ transforms
  (intuitively, $L$ is the boundary between \csssk{k} and $\text{CSS}_{\Delta}$).
  Let $v_L$ (resp. $v'_L$) be the last vertex in $L$ in the path from $v$ (resp. $v'$) to $v_2$.
  Let $v'' = \min\set{v_L, v'_L}$ be the second vertex of $v_L$ and $v'_L$ along $L$.
  Then, $v$ and $v'$ are two incomparable common ancestors of $v_1$ and $v''$ 
  (i.e., $v_L$ in this example) that are both in \csssk{k}.
  This, however, contradicts the inductive hypothesis.
\end{proof}

\begin{figure}[t]
  \centering
  \includegraphics[width = 0.45\textwidth]{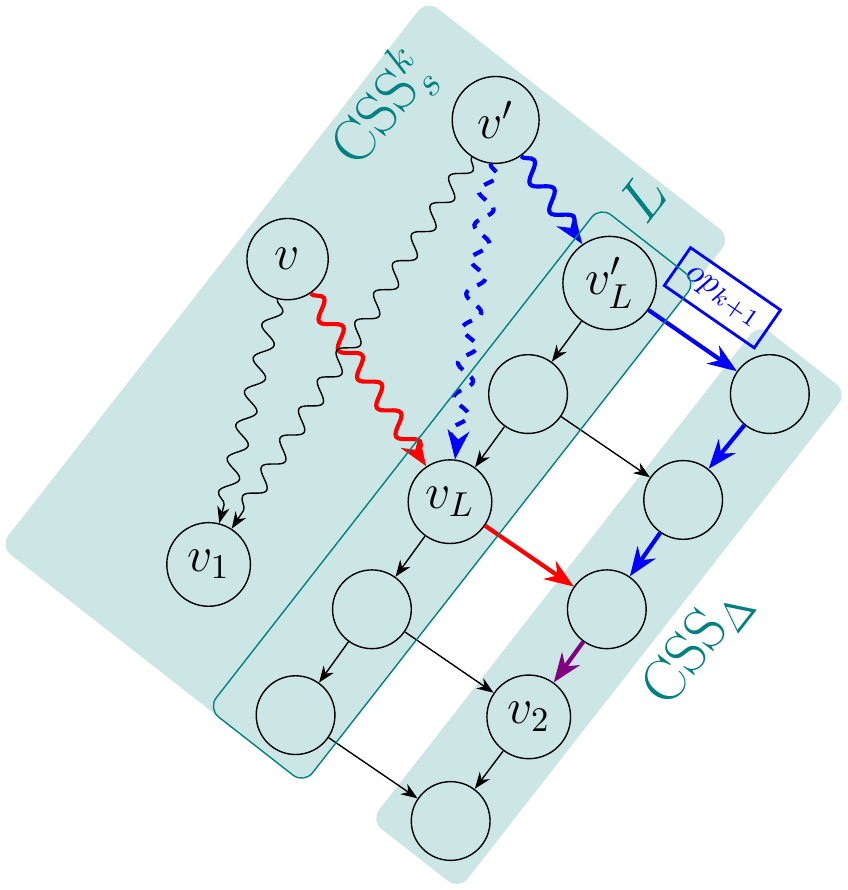}
  \caption{Illustration of proof for Lemma~\ref{lemma:lca}:
    vertices $v$ and $v'$ are two incomparable common ancestors
    of $v_1$ and $v'' = v_{L} \triangleq \min\set{v_{L}, v'_{L}}$ in \csssk{k}.}
  \label{fig:lca}
\end{figure}

\subsection{Proof for Lemma~\ref{lemma:disjoint-paths} (Disjoint Paths)}  \label{ss:proof-disjoint-paths}

\begin{figure}[!t]
  \centering
  \begin{subfigure}[b]{0.35\textwidth}
    \includegraphics[width = \textwidth]{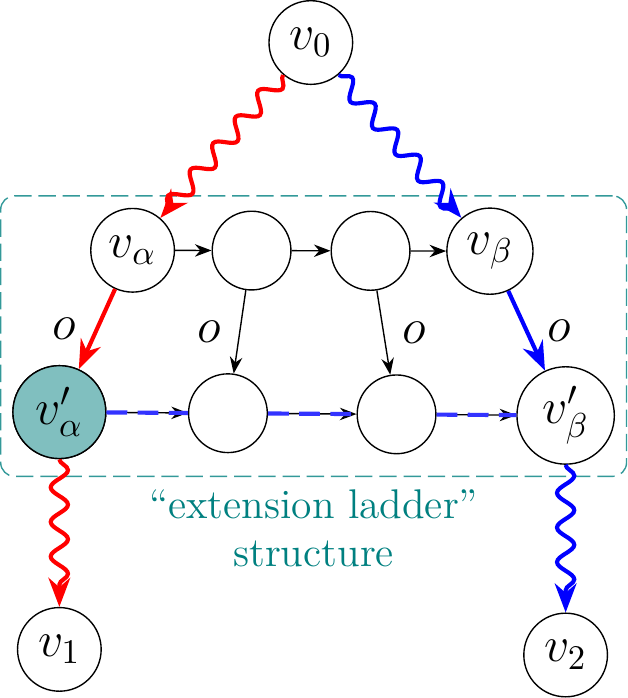}
    \caption{\textsc{Case 2.1:} $v_{\alpha} \xrightarrow{o} v'_{\alpha}$
    and $v_{\beta} \xrightarrow{o} v'_{\beta}$ are in the same ``extension ladder''.}
    \label{fig:disjoint-paths-extension-ladder}
  \end{subfigure}\hfil
  \begin{subfigure}[b]{0.50\textwidth}
    \includegraphics[width = \textwidth]{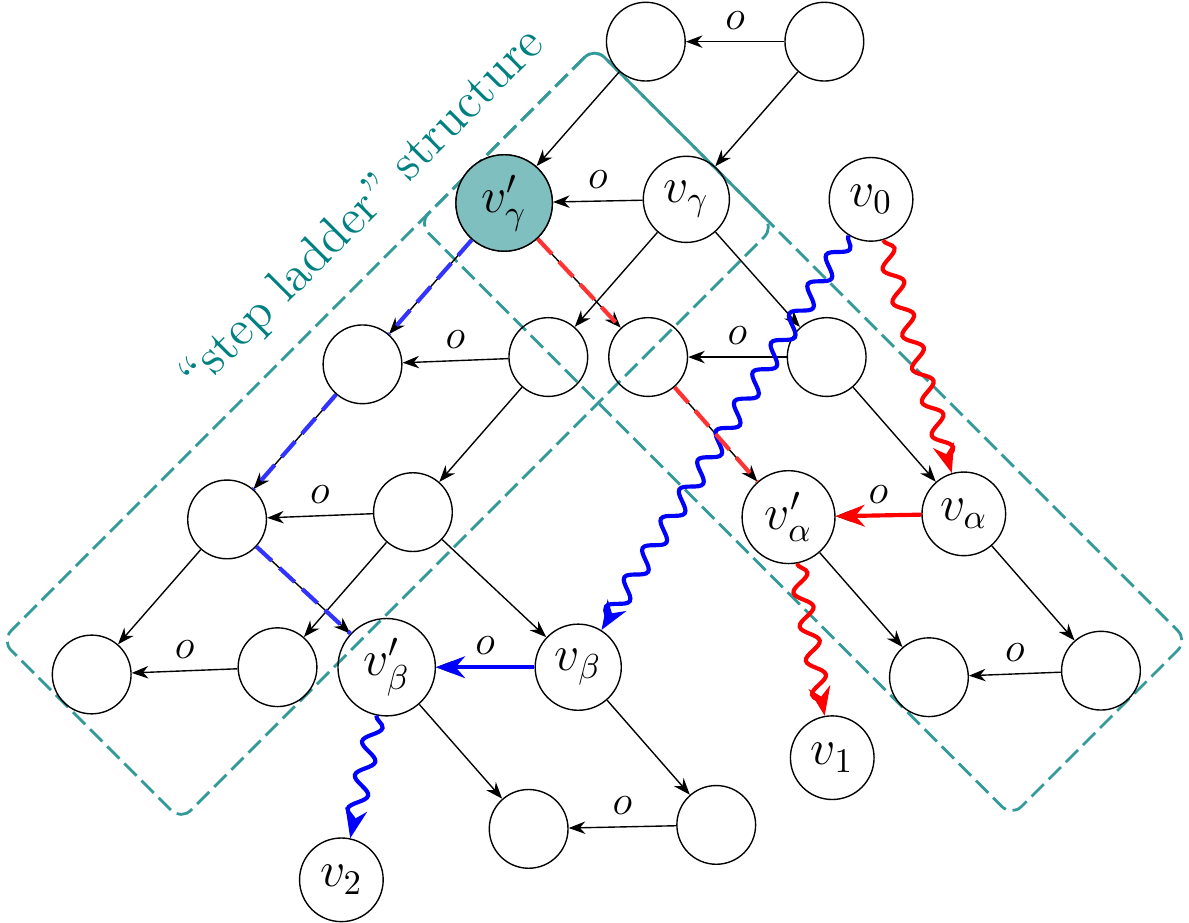}
    \caption{\textsc{Case 2.2:} $v_{\alpha} \xrightarrow{o} v'_{\alpha}$
    and $v_{\beta} \xrightarrow{o} v'_{\beta}$ are in a ``step ladder''.}
    \label{fig:disjoint-paths-step-ladder}
  \end{subfigure}
  \caption{Illustrations of \textsc{Case} 2 of the proof for Lemma~\ref{lemma:disjoint-paths}:
    $v_1$ and $v_2$ are not in the same path from $v_0 = \text{LCA}(v_1, v_2)$.}
  \label{fig:disjoint-paths-ladder}
\end{figure}

\begin{proof}
  We distinguish two cases according to whether $v_1$ and $v_2$ are in the same path 
  from $v_0 = \text{LCA}(v_1, v_2)$ or not.
  
  \textsc{Case} \emph{1: $v_1$ and $v_2$ are in the same path from $v_0 = \emph{LCA}(v_1, v_2)$.}
  In this case, $v_0 = v_1$ or $v_0 = v_2$.
  Therefore, either $O_{v_0 \leadsto v_1}$ or $O_{v_0 \leadsto v_2}$ is empty.
  This lemma obviously holds. 
  
  \textsc{Case} \emph{2: $v_1$ and $v_2$ are not in the same path from $v_0 = \emph{LCA}(v_1, v_2)$.}
  In this case, we prove this lemma by contradiction.
  Suppose that
  \begin{equation*}
    o \in O_{v_0 \leadsto v_1} \cap O_{v_0 \leadsto v_2},
  \end{equation*}
  where $o$ can be either original or transformed (identified by its $oid$).
  As illustrated in Figure~\ref{fig:disjoint-paths-ladder},
  the paths \pathplain{v_0}{v_1} and \pathplain{v_0}{v_2} are now:
  \begin{align*}
    \pathinmath{v_0}{v_1} &= P_{v_0 \leadsto v_{\alpha} \xrightarrow{o} v'_{\alpha} \leadsto v_1}, \\
    \pathinmath{v_0}{v_2} &= P_{v_0 \leadsto v_{\beta} \xrightarrow{o} v'_{\beta} \leadsto v_2}.
  \end{align*}
  In the following, we derive a contradiction that $v_0$ is not the unique LCA of $v_1$ and $v_2$.
  We consider two cases according to how the edges $v_{\alpha} \xrightarrow{o} v'_{\alpha}$
  and $v_{\beta} \xrightarrow{o} v'_{\beta}$ are related via OTs in \csss{}.
  
  \textsc{Case} \emph{2.1: $v_{\alpha} \xrightarrow{o} v'_{\alpha}$ and $v_{\beta} \xrightarrow{o} v'_{\beta}$ are
  in the same ``extension ladder'' structure of OTs.}
  Without loss of generality, we assume that $v'_{\beta}$ is reachable from $v'_{\alpha}$;
  as illustrated in Figure~\ref{fig:disjoint-paths-extension-ladder}.
  In this case, $v'_{\alpha}$ is a lower common ancestor of $v_1$ and $v_2$ than $v_0$.
  This contradicts the condition $\text{LCA}(v_1, v_2) = v_0$.
  
  \textsc{Case} \emph{2.2: $v_{\alpha} \xrightarrow{o} v'_{\alpha}$ and $v_{\beta} \xrightarrow{o} v'_{\beta}$ are
  in a ``step ladder'' structure of OTs.}
  Because all the edges labeled with the same operation $o$ are constructed directly or indirectly
  from the OTs involving the original form of $o$,
  there exists some edge $v_{\gamma} \xrightarrow{o} v'_{\gamma}$ that is in the same ``extension ladder''
  with $v_{\alpha} \xrightarrow{o} v'_{\alpha}$ as well as with $v_{\beta} \xrightarrow{o} v'_{\beta}$;
  as illustrated in Figure~\ref{fig:disjoint-paths-step-ladder}.
  In this case, $v'_{\gamma}$ is a common ancestor of $v_1$ and $v_2$ other than $v_0$.
  This contradicts the condition $\text{LCA}(v_1, v_2) = v_0$.
\end{proof}

\subsection{Proof for Lemma~\ref{lemma:compatible-paths} (Compatible Paths)}    \label{ss:appendix-compatible-paths}

\begin{proof}
  We prove a stronger statement that
  \emph{each pair of vertices consisting of one vertex in $P_{v_0 \leadsto v_1}$
  and the other in $P_{v_0 \leadsto v_2}$ are compatible},
  by mathematical induction on the length $l$ of the path \pathplain{v_0}{v_2}.
  To this end, we first show that

  
  \begin{claim}[One-step Compatibility]
    Suppose that vertices $v$ and $v'$ are compatible.
    Let $v''$ be the next vertex of $v'$ along the edge labeled with operation $op$
    which does not correspond to any element of the list in vertex $v$.
    Then, $v$ and $v''$ are compatible.
  \end{claim}
  
  \begin{subproof}
    Let $C(v,v')$ be the set of common elements of lists in vertices $v$ and $v'$
    and $C(v,v'')$ in vertices $v$ and $v''$.
    By the assumption of this claim, $op$ does not correspond to any element of the list in vertex $v$.
    Therefore, $C(v,v'')$ is a subset of $C(v,v')$.
    Furthermore, the total ordering of elements in $C(v, v'')$ is consistent with that in $C(v, v')$.
  \end{subproof}
  
  \emph{Base Case.} $l = 0$. \pathplain{v_0}{v_2} contains only the vertex $v_0$.
  We shall prove that $v_0$ is compatible with every vertex along \pathplain{v_0}{v_1}.
  This can be done by mathematical induction on the length of \pathplain{v_0}{v_1}
  with the claim above and the fact that \pathplain{v_0}{v_1} is a simple path.
  
  \emph{Inductive Hypothesis.} Suppose that this lemma holds
  when the length of \pathplain{v_0}{v_2} is $l \geq 1$.
  
  \emph{Inductive Step.} We shall prove that the $(l+1)$-\emph{st} vertex, denoted $v_{l+1}$,
  of \pathplain{v_0}{v_2} is compatible with every vertex along \pathplain{v_0}{v_1}.
  This can be done by mathematical induction on the length of \pathplain{v_0}{v_1}
  with the claim above, the fact that \pathplain{v_0}{v_2} is a simple path (for $v_0$ and $v_{l+1}$ being compatible),
  and the fact that \pathplain{v_0}{v_1} and \pathplain{v_0}{v_2} are disjoint.
\end{proof}

\end{document}